\pdfoutput=1
\documentclass[10pt,sigconf]{acmart}
\usepackage{booktabs}
\usepackage[noend]{algpseudocode}
\usepackage{amsmath}
\usepackage{algorithm}
\usepackage{mathtools}
\usepackage{textcomp}
\usepackage{csquotes}
\usepackage{varwidth}
\usepackage{epsfig}
\usepackage{multirow}
\usepackage{amssymb}  %  provides an extended symbol collection

\newcommand{\refsec}[1]{Section~\ref{#1}}

\newcommand{\reffig}[1]{Figure~\ref{#1}}

\newcommand{\qedwhite}{\hfill \ensuremath{\Box}}
\newcommand{\mysubsection}[1]{\vspace{0.9mm} \par\noindent{\bf #1}:}

\newcommand{\onecolfigure}[3]{
\begin{figure}[tbp]
	\centering
	\centerline{\epsfxsize=\columnwidth \epsffile{#1}}
	\vspace*{-1.0em}
	\centerline{\parbox{\columnwidth}{\caption{#2} \label{#3}}}
\end{figure}}

\newcommand{\onecolfigurevcrunch}[3]{
\begin{figure}[tbp]
	\centering
	\centerline{\epsfxsize=\columnwidth \epsffile{#1}}
	\vspace*{-1.0em}
	\centerline{\parbox{\columnwidth}{\caption{#2} \label{#3}}}
	\vspace*{-0.7em}
\end{figure}}

\AtBeginDocument{%
  \providecommand\BibTeX{{%
    \normalfont B\kern-0.5em{\scshape i\kern-0.25em b}\kern-0.8em\TeX}}}

\setcopyright{acmcopyright}
\copyrightyear{2020}
\acmYear{2020}
\acmDOI{10.1145/1122445.1122456}
\acmConference[SIGMOD'20]{Proceedings of the 2020 ACM SIGMOD International Conference on Management of Data}{June 14--19, 2020}{Portland, OR, USA}

\setcopyright{none}
\usepackage{listings}
\renewcommand\footnotetextcopyrightpermission[1]{} % removes footnote
\usepackage{color, colortbl}

\definecolor{dkgreen}{rgb}{0,0.6,0}
\definecolor{gray}{rgb}{0.5,0.5,0.5}
\definecolor{mauve}{rgb}{0.58,0,0.82}
\definecolor{darkblue}{rgb}{0.0, 0.0, 0.4}
\definecolor{cornflowerblue}{rgb}{0.39, 0.58, 0.93}
\definecolor{blizzardblue}{rgb}{0.67, 0.9, 0.93}
\definecolor{columbiablue}{rgb}{0.61, 0.87, 1.0}

\begin{document}

\title[Aggify - Technical Report]{
  Optimizing Cursor Loops in Relational Databases}% \\
  %\large Technical Report}
  
\subtitle {\textbf{Technical Report}}
\titlenote{Extended version of the ACM SIGMOD 2020 paper \enquote{Aggify: Lifting the Curse of Cursor Loops using Custom Aggregates} \cite{Aggify}}

%\title[Aggify]{Optimizing Cursor Loops using Custom Aggregates}

\author{Surabhi Gupta}
\affiliation{Microsoft Research India}
\email{t-sugu@microsoft.com}

\author{Sanket Purandare}
\authornote{Work done during an internship at Microsoft Research India.}
\affiliation{Harvard University}
\email{sanketpurandare@g.harvard.edu}

\author{Karthik Ramachandra}
\affiliation{Microsoft Research India}
\email{karam@microsoft.com}


\begin{abstract}
Loops that iterate over SQL query results are quite common, both in application programs that run outside the DBMS, as well as User Defined Functions (UDFs)  and stored procedures that run within the DBMS. It can be argued that set-oriented operations are more efficient and should be preferred over iteration; but from real world use cases, it is clear that loops over query results are inevitable in many situations, and are preferred by many users. Such loops, known as cursor loops, come with huge trade-offs and overheads w.r.t. performance, resource consumption and concurrency. 

We present Aggify, a technique for optimizing loops over query results that overcomes all these overheads. It achieves this by automatically generating custom aggregates that are equivalent in semantics to the loop. Thereby, Aggify completely eliminates the loop by rewriting the query to use this generated aggregate. This technique has several advantages such as: (i) pipelining of entire cursor loop operations instead of materialization, (ii) pushing down loop computation from the application layer into the DBMS, closer to the data, (iii) leveraging existing work on optimization of aggregate functions, resulting in efficient query plans. We describe the technique underlying Aggify, and present our experimental evaluation over benchmarks as well as real workloads that demonstrate the significant benefits of this technique.
\end{abstract}
%\begin{CCSXML}
%<ccs2012>
%  <concept_id>10010520.10010553.10010562</concept_id>
%  <concept_desc>Computer systems organization~Embedded systems</concept_desc>
%  <concept_significance>500</concept_significance>
%</ccs2012>
%\end{CCSXML}
%
%\ccsdesc[500]{Computer systems organization~Embedded systems}

%\keywords{datasets, neural networks, gaze detection, text tagging}


\maketitle
\pagestyle{empty} 
\section{Introduction}
Since their inception, relational database management systems have emphasized the use of set-oriented operations over iterative, row-by-row operations. SQL strongly encourages the use of set operations and can evaluate such operations efficiently, whereas row-by-row operations are generally known to be inefficient.

However, implementing complex algorithms and business logic in SQL requires %a mindset or approach of thinking and 
decomposing the problem in terms of set-oriented operations. From an application developers' standpoint, this can be fairly hard in many situations. On the other hand, using simple row-by-row operations is often much more intuitive and easier for most developers. As a result, code that iterates over query results and performs operations for every row is extremely common in database applications, as we show in \refsec{subsec:eval-applicability}. 

In fact, the ANSI SQL standard has had the specialized CURSOR construct specifically to enable iteration over query results\footnote{CURSORs have been a part of ANSI SQL at least since SQL-92.} and almost all database vendors support CURSORs.
As a testimonial to the demand for cursors, we note that cursors have been added to procedural extensions of BigData query processing systems such as SparkSQL, Hive and other SQL-on-Hadoop systems~\cite{HPLSQL}. An online search for cursor usage on the popular project hosting service Github yields tens of millions of uses, which gives a sense of their wide usage. Cursors could either be in the form of SQL cursors that can be used in UDFs, stored procedures etc. as well as API such as JDBC that can be used in application programs~\cite{JDBCCURSOR, JDBC}.
%https://docs.microsoft.com/en-us/sql/relational-databases/cursors?view=sql-server-2017

While cursors can be quite useful for developers, they come with huge trade-offs. Primarily, as mentioned earlier, cursors process rows one-at-a-time, and as a result, affect performance severely. Depending upon the cardinality of query results on which they are defined, cursors might materialize results on disk, introducing additional IO and space requirements. Cursors not only suffer from speed problems, but can also acquire locks for a lot longer than necessary, thereby greatly decreasing concurrency and throughput~\cite{CURSE}.

This trade-off has been referred to by many as ``the curse of cursors'' and users are often advised by experts about the pitfalls of using cursors~\cite{CURSE, CURSE2, CURSE3, CURSE4}. A similar trade-off exists in  database-backed applications 
%that use ORMs such as Hibernate~\cite{}, 
where we find code that fetches data from a remote database by submitting a SQL query and iterates over these query results, performing row-by-row operations in application code. 
More generally, imperative programs are known to have serious performance problems when they are executed either in a DBMS or in database-backed applications. This area has been seeing more interest recently and there have been several works that have tried to address this problem~\cite{FroidVldb,Emani2016,Cheung13,cobra18, Park19, sqloop18}.

In this paper, we present Aggify, a technique for optimizing loops over query results. This loop could either be part of application code that runs on the client, or inside the database as part of UDFs or stored procedures.
For such loops, Aggify automatically generates a custom aggregate that is equivalent in semantics to the loop. Then, Aggify rewrites the cursor query to use this new custom aggregate, thereby completely eliminating the loop. 

This rewritten form offers the following benefits over the original program. It avoids materialization of the cursor query results and instead, the entire loop is now a single pipelined query execution. It can now leverage existing work on optimization of aggregate functions~\cite{COH06} and result in efficient query plans. In the context of loops in applications that run outside the DBMS, this can significantly reduce the amount of data transferred between the DBMS and the client. Further, the entire loop computation which ran on the client now runs inside the DBMS, closer to data. Finally, all these benefits are achieved without having to perform intrusive changes to user code. As a result, Aggify is a practically feasible approach with many benefits.

The idea that operations in a cursor loop can be captured as a custom aggregate function was initially proposed by Simhadri et. al.~\cite{uudf}. Aggify is based on this principle which applies to any cursor loop that does not modify database state. We formally prove the above result and also show how the limitations given in~\cite{uudf} can be overcome.
Aggify can seamlessly integrate with existing works on both optimization of database-backed applications~\cite{Emani2016, cobra18} and optimization of UDFs~\cite{FroidVldb, Grust2019}. We believe that Aggify pushes the state of the art in both these (closely related) areas and significantly broadens the scope of prior works. More details can be found in \refsec{sec:relwork}. 
Our key contributions in this paper are as follows.

\begin{enumerate}
	\item We describe Aggify, a language-agnostic technique to optimize loops that iterate over the results of a SQL query. These loops could be either present in applications that run outside the DBMS, or in UDFs/stored procedures that execute inside the DBMS. For such loops, Aggify generates a custom aggregate that is equivalent in semantics to the loop, thereby completely eliminating the loop by rewriting the query to use this generated aggregate.

\item We formally characterize the class of loops that can be optimized by Aggify. In particular, we show that Aggify is applicable to all cursor loops present in SQL User-Defined Functions (UDFs). We also prove that the output of Aggify is semantically equivalent to the input cursor loop.

\item We describe enhancements to the core Aggify technique that expand the scope of Aggify beyond cursor loops to handle iterative FOR loops, and    simplify the generated aggregate using a technique that we call \textit{acyclic code motion}. We also show how Aggify works in conjunction with existing techniques in this space.

\item Aggify has been prototyped on Microsoft SQL Server~\cite{SQLServer}. We discuss the design and implementation of Aggify, and present a detailed experimental evaluation that demonstrates performance gains, resource savings and huge reduction in data movement.
\end{enumerate}

The rest of the paper is organized as follows. We start by motivating the problem in \refsec{sec:motivate} and then provide the necessary background in \refsec{sec:background}. \refsec{sec:aggify} provides an overview of Aggify and presents the formal characterization. \refsec{sec:trans} and \refsec{sec:rewrite} describe the core technique, and \refsec{sec:preserve} reasons about the correctness of our technique. \refsec{sec:enhance} describes enhancements and \refsec{sec:impl} describes the design and implementation. \refsec{sec:eval} presents our experimental evaluation, \refsec{sec:relwork} discusses related work, and \refsec{sec:concl} concludes the paper.

%https://stackoverflow.com/questions/58141/why-is-it-considered-bad-practice-to-use-cursors-in-sql-server
\section{Motivation} \label{sec:motivate}
We now provide two motivating examples and then briefly describe how cursor loops are typically evaluated. 

\subsection{Example: Cursor Loop within a UDF}
\onecolfigure
{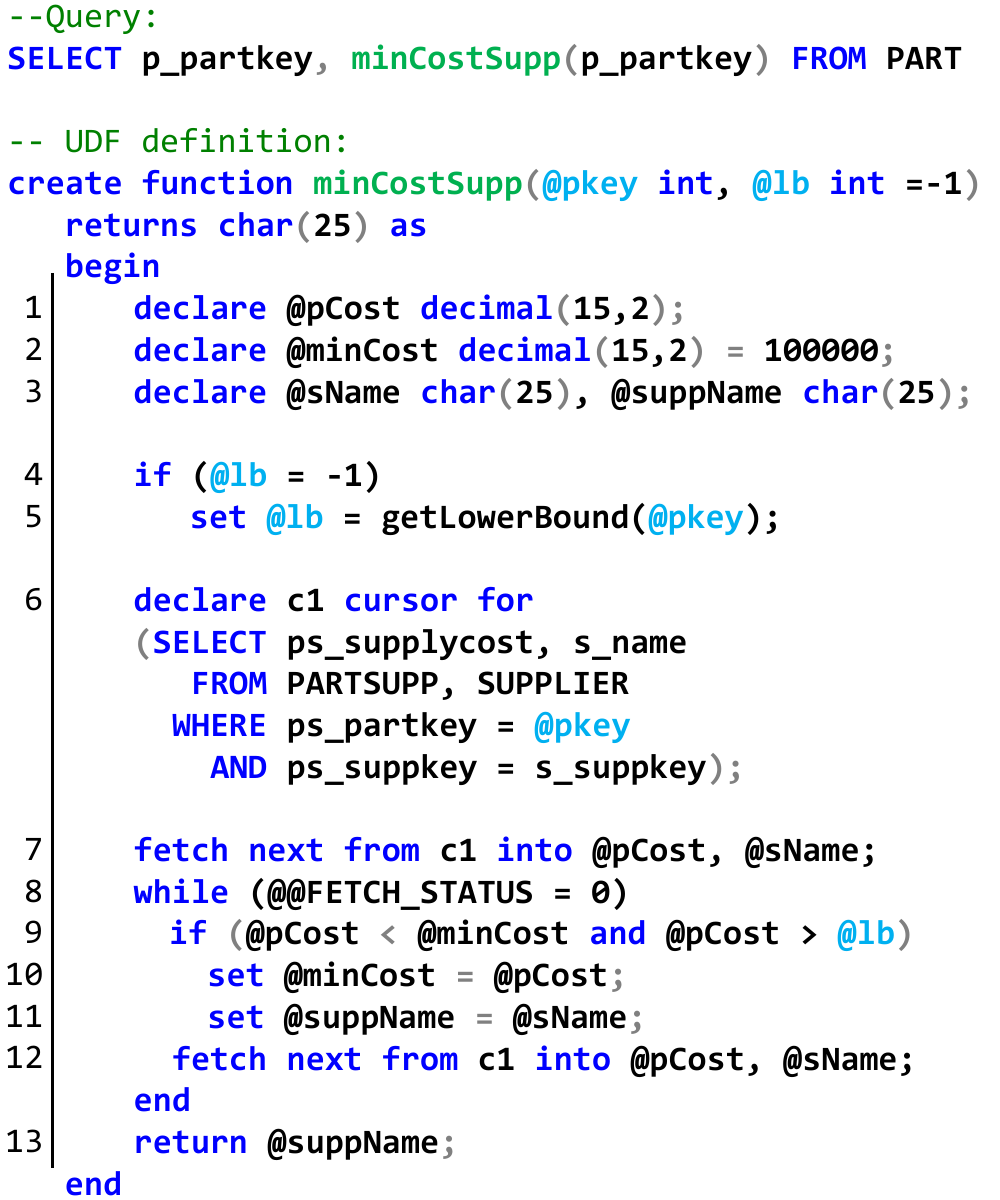}
{Query invoking a UDF that has a cursor loop.}
{fig:udf-example}

Consider a query on the TPC-H schema that is based on query 2 of the TPC-H benchmark, but with a slight modification.
For each part in the PARTS table, this query lists the part identifier (\textit{p\_partkey}) and the name of the supplier that supplies that part with the minimum cost. To this query, we introduce an additional requirement that the user should be able to set a lower bound on the supply cost if required. This lower bound is optional, and if unspecified, should default to a pre-specified value.

Typically, TPC-H query 2 is implemented using a nested subquery. However, another way to implement this is by means of a simple UDF that, given a \textit{p\_partkey}, returns the name of the supplier that supplies that part with the minimum cost. Such a query and UDF (expressed in the T-SQL dialect~\cite{TSQL}) is given in \reffig{fig:udf-example}. As described in~\cite{FroidVldb}, there are several benefits to implement this as a UDF such as reusability, modularity and readability, which is why developers who are not SQL experts often prefer this implementation.

The UDF \textit{minCostSupp} creates a cursor (in line 6) over a query that performs a join between PARTSUPP and SUPPLIER based on the \textit{p\_partkey} attribute. 
Then, it iterates over the query results while 
computing the current minimum cost (while ensuring that it is above the lower bound), and maintaining the name of the supplier who supplies this part at the current minimum (lines 8-12). At the end of the loop, the \textit{@suppName} variable will hold the name of the minimum cost supplier subject to the lower bound constraint, which is then returned from the UDF. Note that for brevity, we have omitted the OPEN, CLOSE and DEALLOCATE statements for the cursor in \reffig{fig:udf-example}; the complete definition of the loop is available in \cite{AGGIFYWL}.

This loop is essentially computing a function that can be likened to \textit{argmin}, which is not a built-in aggregate. This example illustrates the fact that cursor loops can contain arbitrary operations which may not always be expressible using built-in aggregates. 
For the specific cases of functions such as \textit{argmin}, there are advanced SQL techniques that could be used\cite{Emani2016}; however, a cursor loop is the preferred choice for developers who are not SQL experts.
% we can obtain an equivalent SQL query using any of several techniques such as sub-query, a combination of ORDER BY and LIMIT, or using a construct like RANK() if the SQL dialect supports it. However, these are advanced SQL techniques which might be beyond the realm of developers who are not SQL experts. 

\subsection{Example: Cursor Loop in a database-backed Application} \label{java-motivation}
\onecolfigure
{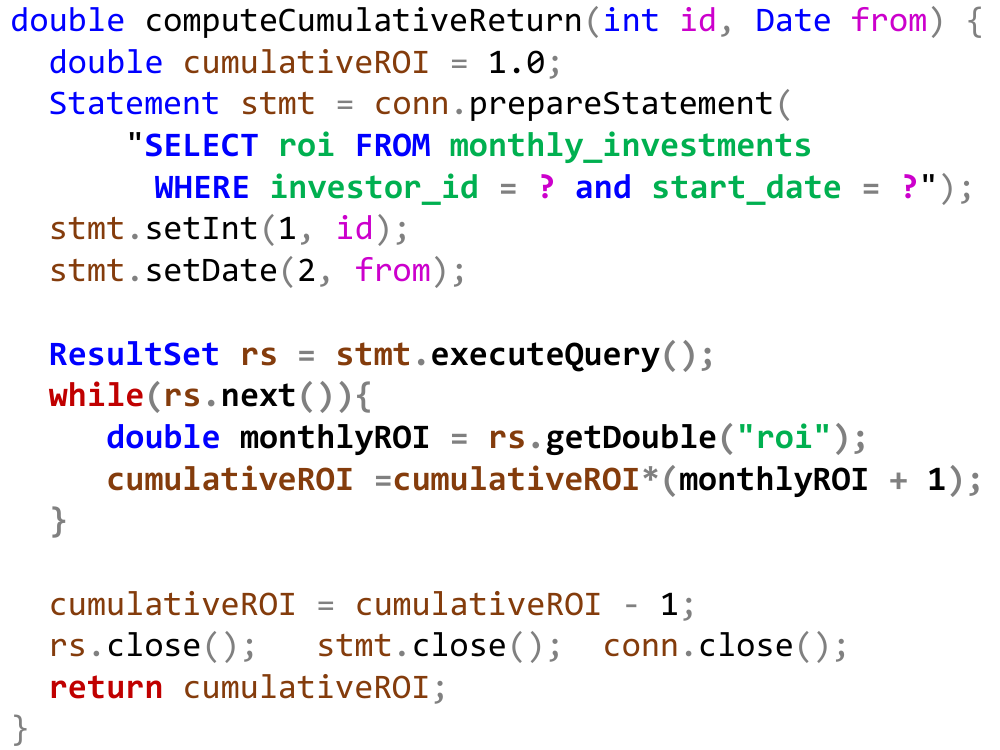}
{Java method computing cumlative ROI for investments using JDBC for database access.}
{fig:jdbc-example}

Consider the scenario of an application that manages investment portfolios for users. \reffig{fig:jdbc-example} shows a Java method from a database-backed application that uses JDBC API~\cite{JDBC} to access a remote database. 
The table \textit{monthly\_investments} includes, among other details, the rate of return on investment (ROI) on a monthly basis. The program first issues a query to retrieve all the monthly ROI values for a particular investor starting from a specified date. Then, it iterates over these monthly ROI values and computes the cumulative rate of return on investment using the time-weighted method\footnote{When the rate of return is calculated over a series of sub-periods of time, the return in each sub-period is based on the investment value at the beginning of the sub-period. 
Assuming returns are reinvested, if the rates over 
n successive time sub-periods are 
$r_{1},r_{2},r_{3},\dots ,r_{n}$, then the cumulative return rate using the time-weighted method is given by~\cite{}: 
$(1+r_{1})(1+r_{2})\dots (1+r_{n})-1$.} and returns the cumulative ROI value.
Observe that this operation is also not expressible using built-in aggregates.
%https://en.wikipedia.org/wiki/Rate_of_return#Returns_over_multiple_periods

% r is actually rate of return

\subsection{Cursor loop Evaluation} \label{subsec:eval}
A cursor is primarily a control structure that enables traversal over the records in the result of a query. They are similar to the programming language concept of iterators. Database systems support different types of cursors such as implicit, explicit, static, dynamic, scrollable, forward-only etc. Our work currently focuses on static, explicit cursors, which is arguably the most widely used type of cursors. 

Cursor loop implementations may vary across database systems, and may also vary based on the type of cursor and other options specified. However, the default behavior in almost all RDBMSs is typically as follows. As part of the evaluation of the cursor declaration (the DECLARE CURSOR statement), the database engine executes the cursor query and materializes the results into a temporary table. This is typically followed by the OPEN CURSOR statement (omitted in \reffig{fig:udf-example} for brevity), which initializes the cursor.
The FETCH NEXT statement moves the cursor and assigns values from the current tuple into local variables. The global variable FETCH\_STATUS indicates whether there are more tuples remaining in the cursor. The body of the WHILE loop is interpreted statement-by-statement, until FETCH\_STATUS indicates that the end of the result set has been reached. Subsequent to this, the cursor is typically closed and deallocated using the CLOSE cursor and DEALLOCATE cursor statements (omitted in \reffig{fig:udf-example} for brevity) that deletes any temporary work tables created by the cursor.

From the above, it is clear that cursor loops can lead to performance issues due to the materialization of the results of the cursor query onto disk, which incurs additional IO and the interpreted evaluation of the loop. This is exacerbated in the presence of large datasets and more so, when invoked repeatedly as in \reffig{fig:udf-example}. The UDF in \reffig{fig:udf-example} is invoked once per part, which means that the cursor query is run multiple times, and temp tables are created and dropped for every run! This is the reason cursors have been referred to as a ``curse'' w.r.t performance~\cite{CURSE}.
\section{Background} \label{sec:background}
We now cover some background material that we make use of in the rest of the paper. %First, we briefly describe custom aggregate functions and their evaluation in DBMSs, and then describe the Data Flow Analysis techniques that are necessary for our work.

\subsection{Custom Aggregate Functions} \label{subsec:cagg}
An aggregate is a function that accepts a collection of values as input and computes a single value by combining the inputs.
Some common operations like \textit{min}, \textit{max}, \textit{sum}, \textit{avg} and \textit{count} are provided by DBMSs as built-in aggregate functions. These are often used along with the GROUP BY operator that supplies a grouped set of values as input.
Aggregates can be deterministic or non-deterministic. Deterministic aggregates return the same output when called with the same input set, irrespective of the order of the inputs. 
All the above-mentioned built-in aggregates are deterministic. Oracle's LISTAGG() is an example of a non-deterministic built-in aggregate function~\cite{LISTAGG}. 

In addition to built-in aggregates, DBMSs allow users to define custom aggregates (also known as User-Defined Aggregates) to implement custom logic. Once defined, they can be used exactly like built-in aggregate functions.
These custom aggregates need to adhere to an aggregation contract~\cite{AGGContract}, typically comprising four methods: \textit{init}, \textit{accumulate}, \textit{terminate} and \textit{merge}. The names of these methods may vary across DBMSs. We now briefly describe this contract.

\begin{enumerate}
\item \textit{Init}: This method is used to initialize variables (fields) that maintain the internal state of the aggregate. It is invoked once for each group aggregated. 

\item \textit{Accumulate}: Defines the main aggregation logic. It is called once for each qualifying tuple in the group being aggregated. It updates the internal state of the aggregate to reflect the effect of the incoming tuple.

\item \textit{Terminate}: Returns the final aggregated value. It might optionally perform some computation as well.

\item \textit{Merge}: This method is optional; it is used in parallel execution of the aggregate to combine partially computed results from different invocations of \textit{Accumulate}.
\end{enumerate}

Out of these 4 methods, the \textit{Merge} method is optional, since it is only necessary in the context of intra-query parallelism. If the query invoking the aggregate function does not use parallelism, the \textit{Merge} method is never invoked. The other 3 methods are mandatory as they are necessary for achieving the aggregation behavior.
The aggregation contract does not enforce any constraint on the order of the input. If order is required, it has to be enforced outside of this contract~\cite{LISTAGG}.

% To illustrate the way custom aggregates are defined in databases, consider the following example of the \textit{prod} user defined aggregate in SQL server which computes the product of the arguments passed. This example is reproduced from~\cite{COH06}.

% \begin{verbatim}
% public class Prod {
%     double val;
    
%     public void Init() {
%         val = 1;
%     }
%     public void Accumulate(double fetchedVal) {
%         val = val * fetchedVal;
%     }
%     public double Terminate() {
%         return val;
%     }
%     public void Merge(Prod other) {
%         val = val * other.val;
%     }
% }
% \end{verbatim}

Several optimizations on aggregate functions have been explored in previous literature~\cite{COH06}. These involve moving the aggregate around joins and allowing them to be either evaluated eagerly or be delayed depending on cost based decisions ~\cite{Yan95eageraggregation}. Duplicate insensitivity and null invariance can also be exploited to optimize aggregates~\cite{Gal01}. There are also well-known techniques that exploit parallelism by partitioning the data, performing local aggregation on each partition and merging the results using global aggregation.

%However, It is important to note that for correctness, only deterministic aggregates should have parallel execution enabled. Hence, the the onus is on the programmer defining the custom aggregate to implement the merge function correctly or decide to disable parallelism based on the operations being done inside the aggregate. 

\subsection{Data Flow Analysis}

\onecolfigurevcrunch
{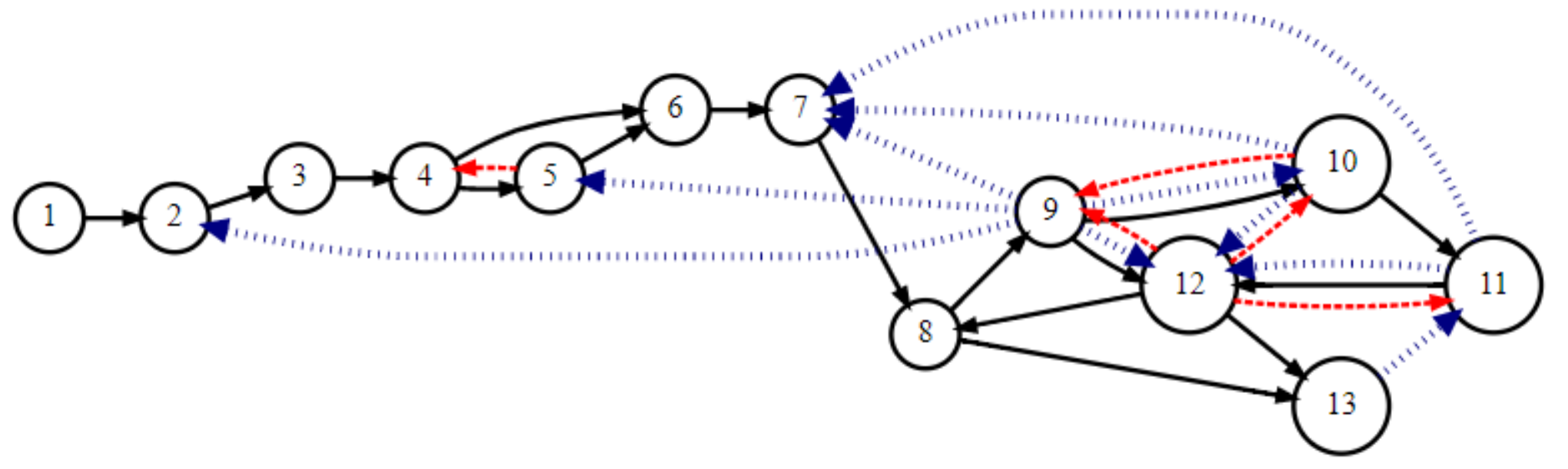}
{Control Flow Graph for the UDF in \reffig{fig:udf-example}, augmented with data dependence edges.}
{fig:CFG_ex1}

We now briefly describe the data structures and static analysis techniques that we make use of in this paper. The material in this section is mainly derived from~\cite{KENNEDYBOOK,Aho06, DFABOOK} and we refer the readers to these for further details. 

Data flow analysis is a program analysis technique that is used to derive information about the run time behaviour of a program~\cite{KENNEDYBOOK,Aho06,DFABOOK}. 
{\color{black} The \textit{Control Flow Graph (CFG)} of a program is a directed graph where vertices represent basic blocks (a straight line code sequence with no branches) and edges represent transfer of control between basic blocks during execution.
The \textit{Data Dependence Graph} (DDG) of a program is a directed multi-graph in which program statements are nodes, and the edges represent data dependencies between  statements. Data dependencies could be of different kinds -- Flow dependency (read after write), Anti-dependency (write after read), and Output dependency (write after write).  The entry and exit point of any node in the CFG is denoted as a \textit{program point}.

\reffig{fig:CFG_ex1} shows the CFG for the UDF in \reffig{fig:udf-example}. Here we consider each statement to be a separate basic block. The CFG has been augmented with data dependence edges where the dotted (blue) and dashed (red) arrows respectively indicate flow and anti dependencies. We use this augmented CFG (sometimes referred to as the Program Dependence Graph or PDG~\cite{PDG}) as the input to our technique.}

\subsubsection{Framework for data flow analysis}
A data-flow value for a program point is an abstraction of the set of all possible program states that can be observed for that point. For a given program entity $e$, such as a variable or an expression, data flow analysis of a program involves (i) discovering the effect of individual program statements on $e$ (called local data flow analysis), and (ii) relating these effects across statements in the program (called global data flow analysis) by propagating data flow information from one node to another. 
The relationship between local and global data flow information is captured by a system of data flow equations. The nodes of the CFG are traversed and these equations are iteratively solved until the system reaches a fixpoint. The results of the analysis can then be used to infer information about the program entity $e$.

\subsubsection{UD and DU Chains} \label{subsubsec:duchain}
When a variable $v$ is on the LHS of an assignment statement $S$, $S$ is known as a \textit{Definition} of $v$. When a variable $v$ is on the RHS of an assignment statement $S$, $S$ is known as a \textit{Use} of $v$.  
A Use-Definition (UD) Chain is a data structure that consists of a use U, of a variable, and all the definitions D, of that variable that can reach that use without any other intervening definitions. A counterpart of a UD Chain is a Definition-Use (DU) Chain which consists of a definition D, of a variable and all the uses U, reachable from that definition without any other intervening definitions. These data structures are created using data flow analysis. 

\subsubsection{Reaching definitions analysis} \label{subsubsec:reach}
This analysis is used to determine which definitions reach a particular point in the code. A definition \textit{D} of a variable reaches a program point \textit{p} if there exists a path leading from \textit{D} to \textit{p} such that \textit{D} is not overwritten (killed) along the path. The output of this analysis can be used to construct the UD and DU chains which are then used in our transformations.
For example, in \reffig{fig:udf-example}, consider the use of the variable \textit{@lb} inside the loop (line 9). There are at least two definitions of \textit{@lb} that reach this use. One is the the initial assignment of \textit{@lb} to -1 as a default argument, and the other is assignment on line 5.

% The set of reaching definitions going into basic block \textit{b} (IN[b])
% is the union of all the reaching definitions coming out from 
% \textit{b}'s predecessor blocks in the control flow graph. The out-state (OUT[b]) of a basic block \textit{b} consists of all the reaching definitions coming out of its predecessor blocks minus those reaching definitions whose variable is killed by \textit{b} (overwritten)
%  plus any new definitions generated within \textit{b}. 

\subsubsection{Live variable analysis} 
This analysis is used to determine which variables are \textit{live} at each program point. A variable is said to be \textit{live} at a point if it has a subsequent use before a re-definition.
For example, consider the variable \textit{@lb} in \reffig{fig:udf-example}. This variable is \textit{live} at every program point in the loop body. But at the end of the loop, it is no longer \textit{live} as it is never used beyond that point. In the function \textit{minCostSupp}, the only variable that is \textit{live} at the end of the loop is \textit{@suppName}. We will use this information in Aggify as we shall show in \refsec{sec:trans}.

% Here, the IN state of a block is the set of variables that are live at the start of the block and OUT state is the set of variables which are live at the end of the block. The OUT state of a block \textit{b} is computed by taking union over the IN states of it's successor basic blocks. The IN state of a block is \textit{b} is computed by taking the variables that are live at the beginning of each successor block of \textit{b}, minus the variables which are written to in \textit{b} plus the set of variables that are used in this block before any assignment. This process is done iteratively until the IN and OUT states of blocks stop changing.  

\section{Aggify Overview} \label{sec:aggify}
Aggify is a technique that offers a solution to the limitations of cursor loops described in \refsec{subsec:eval}. It achieves this goal by replacing the entire cursor loop with an equivalent SQL query that invokes a custom aggregate that is systematically constructed.
Performing such a rewrite that guarantees equivalence in semantics is nontrivial.
The key challenges involved here are the following. The body of the cursor loop could be arbitrarily complex, with cyclic data dependencies and complex control flow. The query on which the cursor is defined could also be arbitrarily complex, having subqueries, aggregates and so on. Furthermore, the UDF or stored procedure that contains this loop might define variables that are used and modified within the loop. 

In the subsequent sections, we show how Aggify achieves this goal such that the rewritten query is semantically equivalent to the cursor loop. Aggify primarily involves two phases.
The first phase is to construct a custom aggregate by analyzing the loop (described in \refsec{sec:trans}). Then, the next step is to rewrite the cursor query to make use of the custom aggregate and removing the entire loop (described in \refsec{sec:rewrite}). 

\subsection{Applicability} \label{subsec:applicability}
Before delving into the technique, we formally characterize the class of cursor loops that can be transformed by Aggify and specify the supported operations inside such loops.

\begin{definition}
A \textit{Cursor Loop} (CL) is defined as a tuple ($Q, \Delta$) where $Q$ is any SQL SELECT query and $\Delta$ is a program fragment that can be evaluated over the results of $Q$, one row at a time.
\end{definition}

Observe that in the above definition, the body of the loop ($\Delta$) is neither specific to a programming language nor to the execution environment. The loop can be either implemented using procedural extensions of SQL, or using other general purpose programming languages such as Java. This definition therefore encompasses the loops shown in Figures \ref{fig:udf-example} and \ref{fig:jdbc-example}. 
In general, the statements in the loop can include arbitrary operations that may even mutate the persistent state of the database. However, such loops cannot be transformed by Aggify, since aggregates by definition cannot modify database state. We now state the theorem that defines the applicability of Aggify.

\begin{theorem} \label{th:app}
Any cursor loop CL($Q, \Delta$) that does not modify the persistent state of the database can be equivalently expressed as a query $Q'$ that invokes a custom aggregate function $Agg_{\Delta}$.
\end{theorem}
\begin{proof}
We prove this theorem in three steps. 
\begin{enumerate}
    \item We describe (in \refsec{sec:trans}) a technique to systematically construct a custom aggregate function $Agg_{\Delta}$ for a given cursor loop CL($Q, \Delta$).
    \item We present (in \refsec{sec:rewrite}) the  rewrite rule that can be used to rewrite the cursor loop as a query $Q'$ that invokes $Agg_{\Delta}$
    \item We show (in \refsec{sec:preserve}) that the rewritten query $Q'$ is semantically equivalent to the cursor loop CL($Q, \Delta$).
\end{enumerate}
By steps (1), (2), and (3), the theorem follows.
\end{proof}

Observe that Theorem~\ref{th:app} encompasses a fairly large class of loops encountered in reality. More specifically, this covers all cursor loops present in user-defined functions (UDFs). This is because UDFs by definition are not allowed to modify the persistent state of the database. As a result, all cursor loops inside such UDFs can be rewritten using Aggify. Note that this theorem only states that a rewrite is possible; it does not necessarily imply that such a rewrite will always be more efficient. There are several factors that influence the performance improvements due to this rewrite, and we discuss them in our experimental evaluation (\refsec{sec:eval}).

\subsection{Supported operations} \label{subsec:support}
We support all operations inside a loop body that are admissible inside a custom aggregate. The exact set of operations supported inside a custom aggregate varies across DBMSs, but in general, this is a very broad set which includes standard procedural constructs such as variable declarations, assignments, conditional branching, nested loops (cursor and non-cursor) and function invocations. All scalar and table/collection data types are supported. The formal language model that we support is given below.
{\small
\begin{align*}
    expr ::=& \text{ \textit{Constant} } | \text{ \textit{var} } | \text{ \textbf{Func(...)} } | \text{ \textbf{Query(...)} } \\
                \vspace{-5mm}
            &| \text{ $\neg$ expr } | \text{ \textit{expr1} \textbf{op} \textit{expr2} }\\
            \vspace{-5mm}
    op  ::=& \text{ + $|$ - $|$ * $|$ / $|$ \textless $|$ \textgreater $|$ ... }\\
    Stmt ::=& \text{ \textit{skip} } | \text{ \textit{Stmt; Stmt} } | \text{ \textit{var} := \textit{expr} } \\
            &| \text{ \textbf{if} \textit{expr} \textbf{then} \textit{Stmt} \textbf{else} \textit{Stmt}} \\
            &| \text{ \textbf{while} \textit{expr} \textbf{do} \textit{Stmt}} \\
            &| \text{ \textbf{try} \textit{Stmt} \textbf{catch} \textit{Stmt}} \\
    Program ::=& \text{ \textit{Stmt} }
\end{align*}}
Nested cursor loops are supported as described in \refsec{subsubsec:nested}. SQL SELECT queries inside the loop are fully supported. DML operations (INSERT, UPDATE, DELETE) on local table variables or temporary tables or collections are supported.  We can also support operations having side-effects (such as writing to a file) if the DBMS allows these operations inside a custom aggregate. Exception handling code (TRY...CATCH) can also be supported. Operations that may change the persistent state of the database (DML statements against persistent tables, transactions, configuration changes etc.) are not supported. Unconditional jumps such as BREAK and CONTINUE are not directly supported\footnote{This is not a fundamental limitation, as such loops can be preprocessed and rewritten without using unconditional jumps.}.
We now describe the core Aggify technique in detail.
% In the following sections, we describe the technique at the core of Aggify. We first describe the basic technique, and then present some enhancements in \refsec{sec:enhance}.

\section{Custom Aggregate Construction} \label{sec:trans}
\onecolfigure
{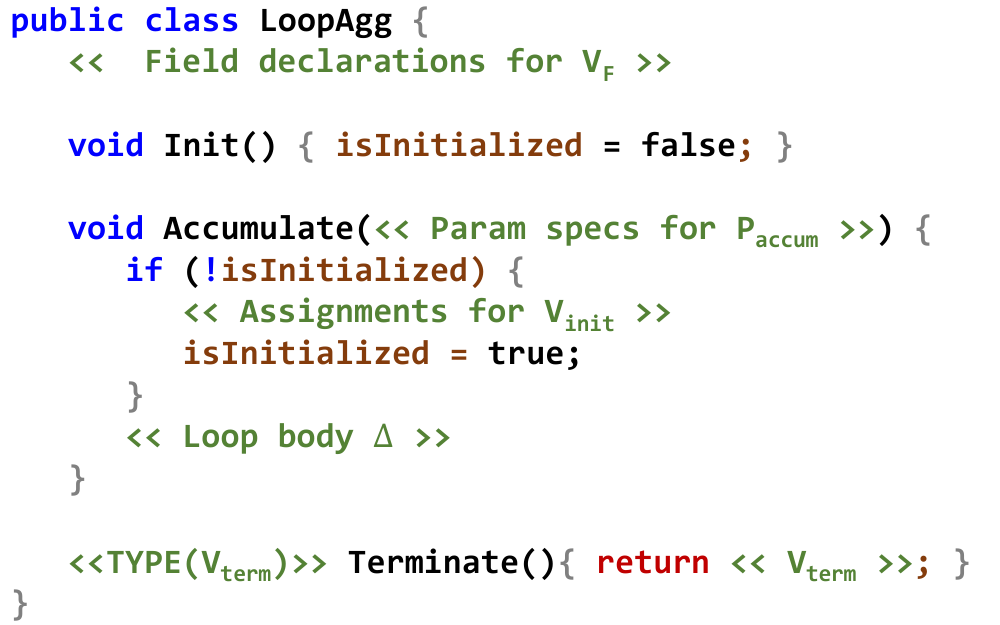}
{Template for the custom aggregate.}
{fig:agg-template}

\onecolfigure
{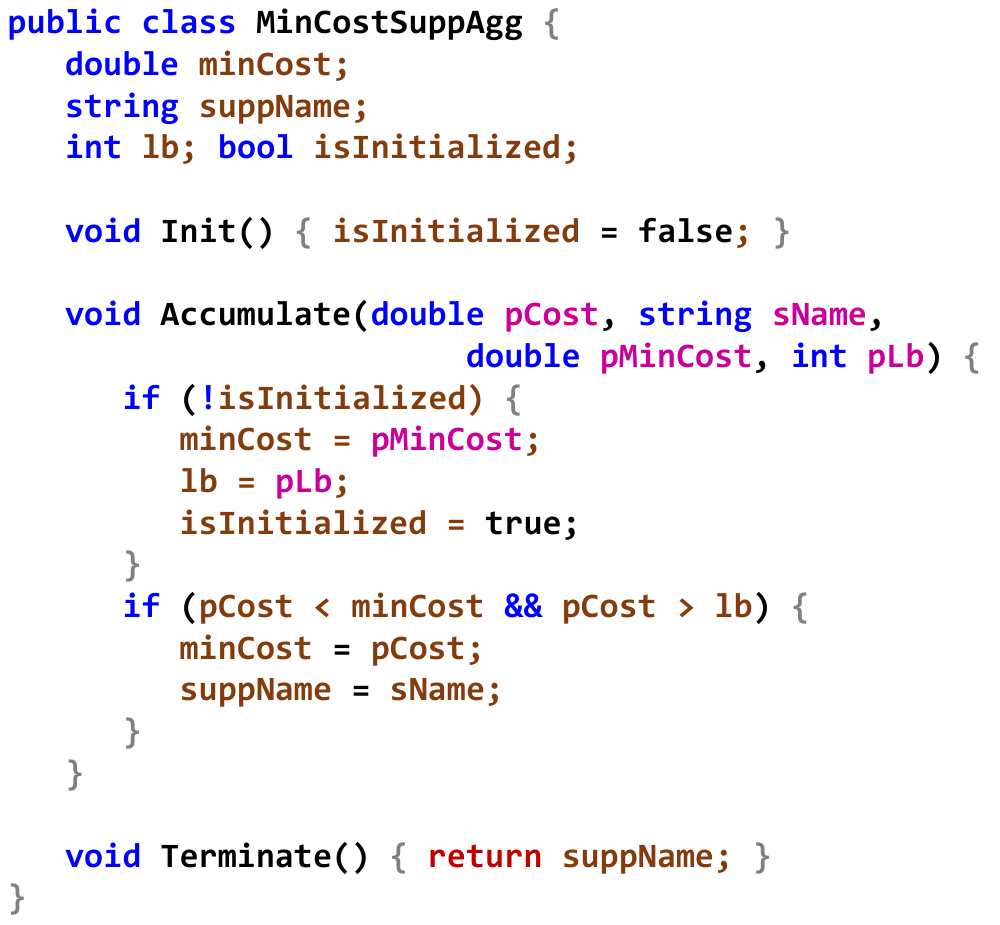}
{Custom aggregate for the loop in \reffig{fig:udf-example} constructed by Aggify.}
{fig:udf-agg}

\onecolfigure
{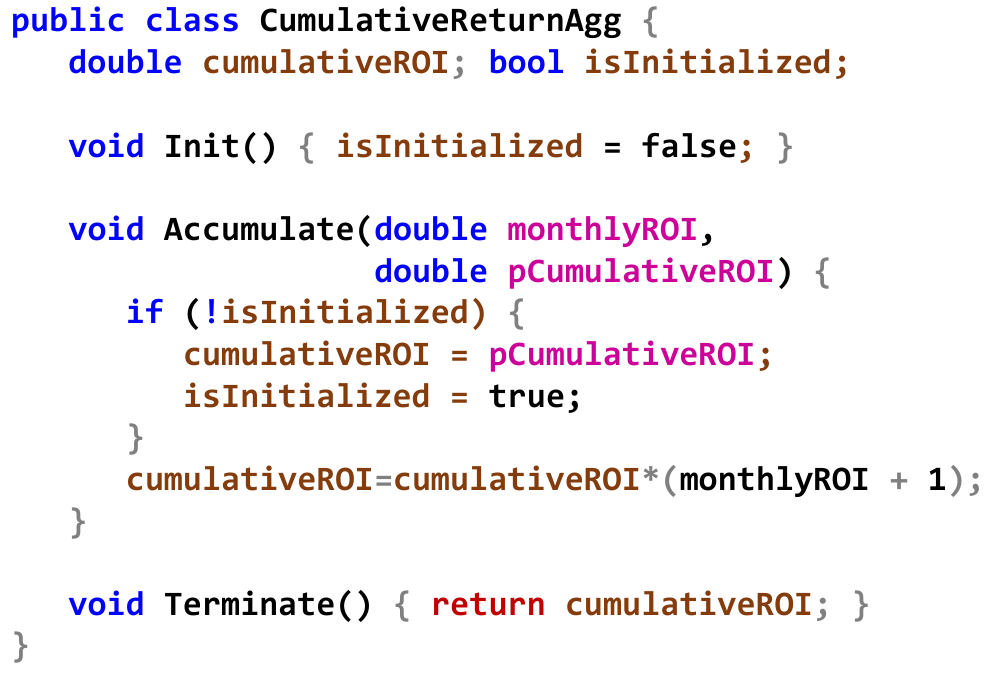}
{Custom aggregate for the loop in \reffig{fig:jdbc-example} constructed by Aggify.}
{fig:java-agg}

Given a cursor loop (Q, $\Delta$) our goal is to construct a custom aggregate that is equivalent to the body of the loop, $\Delta$.
As explained in \refsec{subsec:cagg}, we use the aggregate function contract involving the 3 mandatory methods -- Init, Accumulate and Terminate -- as the target of our construction. 
%Write about merge and that it is part of future work.
Constructing such a custom aggregate involves specifying its signature (return type and parameters), fields and constructing the three method definitions. 
\reffig{fig:agg-template} shows the template that we start with. The patterns $<< >>$ in \reffig{fig:agg-template} (shown in green) indicate `holes' that need to be filled with code fragments inferred from the loop.
We now show how to construct such an aggregate and illustrate it using the examples from \refsec{sec:motivate}.
Figures \ref{fig:udf-agg} and \ref{fig:java-agg} show the definition of the custom aggregate for the loops in Figures \ref{fig:udf-example} \ref{fig:jdbc-example} respectively. We use the syntax similar to that of Microsoft SQL Server to illustrate these examples. 

\subsection{Fields} \label{subsec:fields}
Conceptually, all variables live at the beginning of the loop can be made fields of the aggregate. While this is not incorrect, it is unnecessary. Therefore we identify a minimal set of fields as follows. Consider the set $V_{\Delta}$ of all variables referenced in the loop body $\Delta$. Let $V_{fetch}$ be the set of variables assigned in the \textit{FETCH} statement, and let $V_{local}$ be the set of variables that are local to the loop body (i.e they are declared within the loop body and are not \textit{live} at the end of the loop.)
The set of variables $V_F$ defined as fields of the custom aggregate is given by the equation:
\begin{equation} \label{eq:fields}
V_F = (V_{\Delta} - (V_{fetch} \cup V_{local})) \cup \{isInitialized\}.
\end{equation}
We have additionally added a variable called \textit{isInitialized} to the field variables set $V_F$. This boolean field is necessary for keeping track of field initialization, and will be described in \refsec{subsec:init}.
For all variables in $V_F$, we place a field declaration statement in the custom aggregate class.

\mysubsection{Illustrations} For the loop in \reffig{fig:udf-example}, 
\vspace{-2mm}
\begin{align*}
    V_{\Delta} &= \{pCost, minCost, lb, suppName, sName\} \\
    V_{fetch} &= \{pCost, sName\} \\
    V_{local} &= \{\} \\
\text{Therefore, using } & \text{\refeq{eq:fields}, we get}\\
    V_{F} &= \{minCost, lb, suppName, isInitialized\} \qedwhite
\end{align*}

For application programs such as the one in \reffig{fig:jdbc-example} that use a data access API like JDBC, the attribute accessor methods (e.g. \textit{getInt(), getString()} etc.) on the \textit{ResultSet} object are treated analogous to the \textit{FETCH} statement. Therefore, local variables to which \textit{ResultSet} attributes are assigned form a part of the $V_{fetch}$ set.
For the loop in \reffig{fig:jdbc-example}, 
\vspace{-2mm}
\begin{align*}
    V_{\Delta} &= \{cumulativeROI, monthlyROI\} \\
    V_{fetch} &= \{monthlyROI\} \\
    V_{local} &= \{\} \\
\text{Therefore, using } & \text{Equation~\refeq{eq:fields}, we get}\\
    V_{F} &= \{cumulativeROI, isInitialized\}  \qedwhite
\end{align*}

\subsection{Init()} \label{subsec:init}
The implementation of the \textit{Init()} method is very simple. We just add a statement that assigns the boolean field \textit{isInitialized} to \textit{false}. Initialization of field variables is deferred to the \textit{Accumulate()} method for the following reason.
The \textit{Init()} does not accept any arguments. Hence if field initialization statements are placed in \textit{Init()}, they will have to be restricted to values that are statically determinable~\cite{uudf}. This is because these values will have to be supplied at aggregate function creation time. In practice it is quite likely that these values are not statically determinable. This could be because (a) they are not compile-time constants but are variables that hold a value at runtime, or (b) there are multiple definitions of these variables that might reach the loop, due to presence of conditional assignments. 

Consider the loop of \reffig{fig:udf-example}. Based on \refeq{eq:fields}, we have determined that the variable \textit{@lb} has to be a field of the custom aggregate. Now, we cannot place the initialization of \textit{@lb} in \textit{Init()} because there is no way to determine the initial value of \textit{@lb} at compile-time using static analysis of the code. This was a restriction in~\cite{uudf} which we overcome by deferring field initializations to \textit{Accumulate()}.

\mysubsection{Illustrations} The \textit{Init()} method is identical in both Figures \ref{fig:udf-agg} and \ref{fig:java-agg}, having an assignment of \textit{isInitialized} to \textit{false}.

\subsection{Accumulate()} \label{subsec:accum}
In a custom aggregate, the \textit{Accumulate()} method encapsulates the important computations that need to happen. We now describe how to systematically construct the parameters and the method definition of \textit{Accumulate()}.

\subsubsection{Parameters.}
Let $P_{accum}$ denote the set of parameters which is identified as \textit{the set of  variables that are used inside the loop body and have at least one reaching definition outside the loop}. The set of candidate variables is computed using the results of reaching definitions analysis (\refsec{subsubsec:reach}). 
More formally, let $V_{use}$ be the set of all variables used inside the loop body. 
For each variable $v$ $\in$ $V_{use}$,
let $U_{CL}(v)$ be the set of all uses of $v$ inside the cursor loop CL.
Now, for each use $u$ $\in$ $U_{CL}(v)$, 
let RD($u$) be the set of all definitions of $v$ that reach the use $u$.
We define a function $R(v)$ as follows.
\begin{equation} \label{eq:rv}
  R(v)=\begin{cases}
    1, & \text{if } \exists d \in RD(u) \mid d \text{ is not in the loop}.\\
    0, & \text{otherwise}.
  \end{cases}
\end{equation}

% Then the set ReachingDefs$(V_{use})$ contains reaching definitions of all variables in $V_{use}$.
% We now define a function $R(v)$ as follows.

% \begin{math}
% \forall \text{ }D(v) \in \text{ReachingDefs}(V_{use}),
% \end{math}
% \begin{equation} \label{eq:rv}
%   R(v)=\begin{cases}
%     1, & \text{if } D(v) \text{ is not in the loop}.\\
%     0, & \text{otherwise}.
%   \end{cases}
% \end{equation}

Checking if a definition $d$ is in the loop or not is a simple set containment check. Using \refeq{eq:rv}, we define $P_{accum}$, the set of parameters for \textit{Accumulate()} as follows.
\begin{align} \label{eq:paccum}
    P_{accum} = \{v \mid v \in V_{use} \land R(v) == 1\}
\end{align}

\subsubsection{Method Definition.}
There are two blocks of code that form the definition of \textit{Accumulate()} --  field initializations and the loop body block. The set of fields $V_{init}$ that need to be initialized is given by the below equation.
\begin{equation} \label{eq:vinit}
V_{init} = P_{accum} - V_{fetch}    
\end{equation}

As mentioned earlier, the boolean field \textit{isInitialized} denotes whether the fields of this class are initialized or not. The first time accumulate is invoked for a group, \textit{isInitialized} is false and hence the fields in $V_{init}$ are initialized. During subsequent invocations, this block is skipped as \textit{isInitialized} would be true. Following the initialization block, the entire loop body $\Delta$ is appended to the definition of \textit{Accumulate()}. 
\mysubsection{Illustrations}
For the loop in \reffig{fig:udf-example}:
\begin{align*}
    P_{accum} &= \{ pCost, sName, pMinCost, pLb \} \\
    V_{init} &= \{ minCost, lb \}
\end{align*}

For the loop in \reffig{fig:jdbc-example}, $P_{accum}$ and $V_{init}$ are as follows:
\begin{align*}
    P_{accum} &= \{ monthlyROI, cumulativeROI \} \\
    V_{init} &= \{ cumulativeROI \}
\end{align*} 
The \textit{Accumulate()} method in Figures \ref{fig:udf-agg} and \ref{fig:java-agg} are constructed based on the above equations as per the template in \reffig{fig:agg-template}.

\subsection{Terminate()} \label{subsec:terminate}
This method returns a tuple of all the field variables ($V_F$) that are \textit{live} at the end of the loop. The set of candidate variables $V_{term}$ are identified by performing a liveness analysis for the module enclosing the cursor loop (e.g. the UDF that contains the loop).
The return type of the aggregate is a tuple where each attribute corresponds to a variable that is live at the end of the loop. The tuple datatype can be implemented using User-Defined Types in most DBMSs. 

\mysubsection{Illustrations}
For the loop in \reffig{fig:udf-example}, $V_{term} = \{ suppName \}$, and for the loop in \reffig{fig:jdbc-example}, $V_{term} = \{ cumulativeROI \}$. For simplicity, since these are single-attribute tuples, we avoid using a tuple and use the type of the attribute as the return type of \textit{Terminate()}.

\section{Query Rewriting}\label{sec:rewrite}
\onecolfigure
{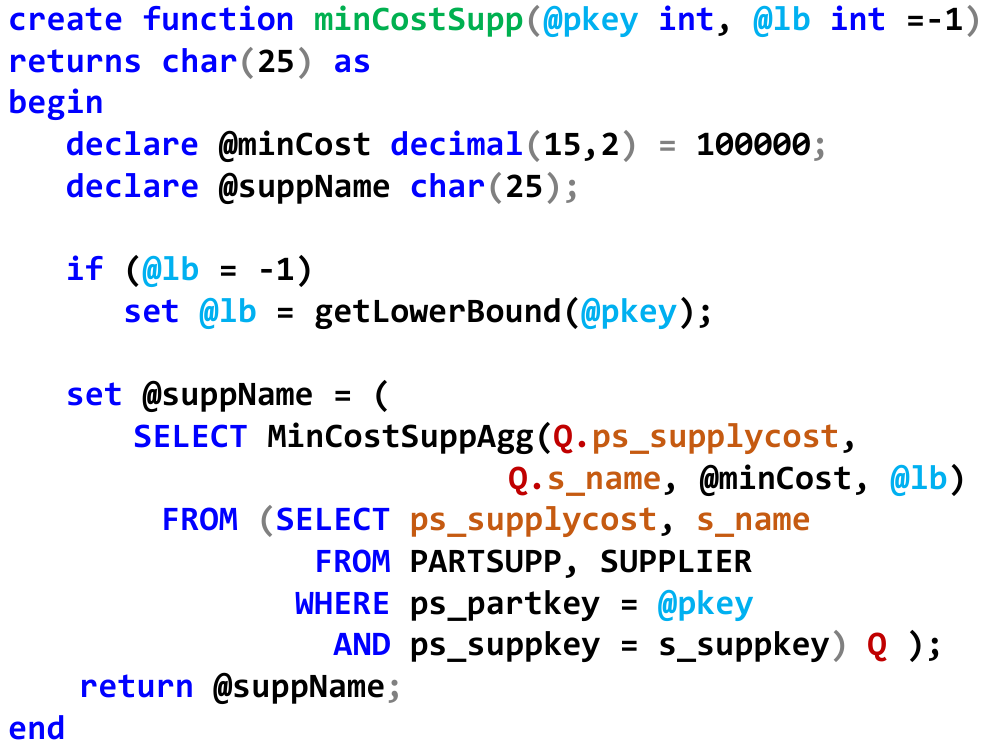}
{The UDF in \reffig{fig:udf-example} rewritten using Aggify.}
{fig:udf-trans}
\onecolfigure
{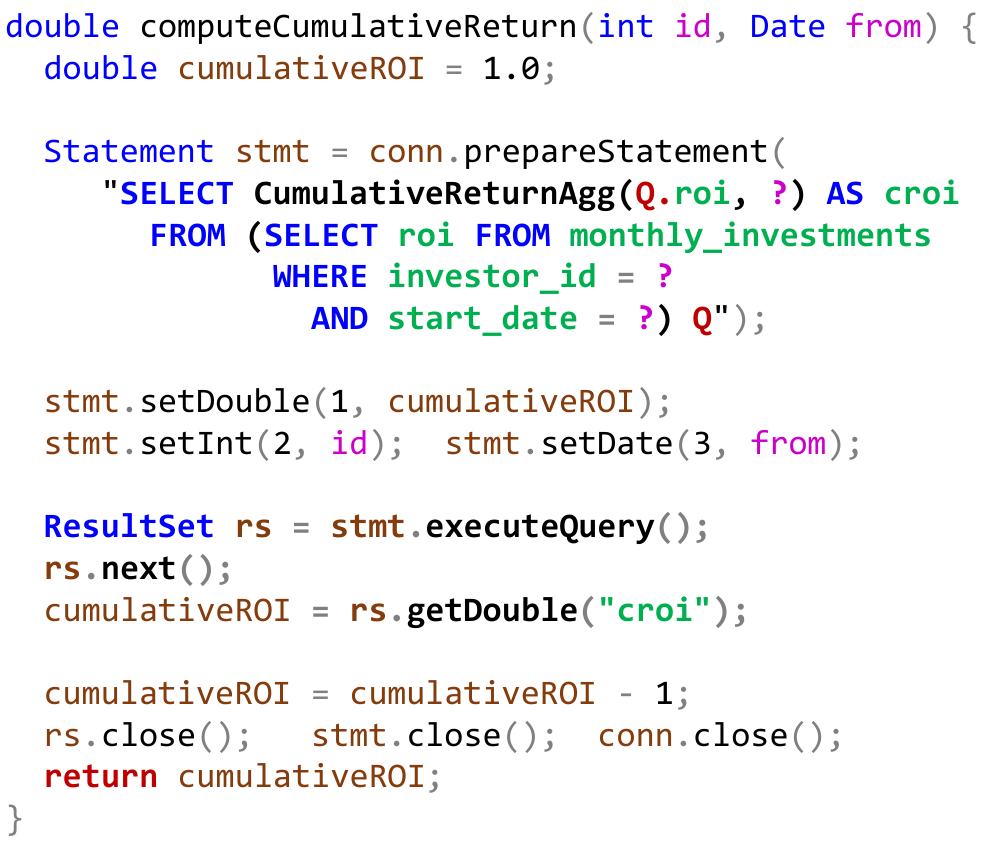}
{The Java method in \reffig{fig:jdbc-example} rewritten using Aggify.}
{fig:jdbc-trans}

For a given cursor loop (Q, $\Delta$), once the custom aggregate $Agg_{\Delta}$ has been created, the next task is to remove the loop altogether and rewrite the query $Q$ into $Q'$ such that it invokes this custom aggregate instead.
Note that $Q$ might be arbitrarily complex, and may contain other aggregates (built-in or custom), GROUP BY, sub-queries and so on.
Therefore, Aggify constructs $Q'$ without modifying $Q$ directly, but by composing $Q$ as a nested sub-query. In other words, Aggify introduces an aggregation on top of $Q$ that contains an invocation to $Agg_{\Delta}$. Note that $Agg_{\Delta}$ is the only attribute that needs to be projected, as it contains all the loop variables that are \textit{live}. In relational algebra, this rewrite rule can be represented as follows:
\begin{equation} \label{eq:rule1a}
    \text{Loop}(Q,\Delta) \implies \mathcal{G}_{\text{Agg}_{\Delta}(P_{\textit{accum}}) \text{ as aggVal}}(Q)
\end{equation}

Note that the parameters to $Agg_{\Delta}$ are the same as the parameters to the \textit{Accumulate()} method ($P_{accum}$). These are either attributes that are projected from $Q$ or variables that are defined earlier. The return value of $Agg_{\Delta}$ (aliased as \textit{aggVal}) is a tuple from which individual attributes can be extracted. The details are specific to SQL dialects. 

\mysubsection{Illustration}
\reffig{fig:udf-trans} shows the output of rewriting the UDF in \reffig{fig:udf-example} using Aggify. Observe the statement that assigns to the variable \textit{@suppName} where the R.H.S is a SQL query. 
For the loop in \reffig{fig:udf-example}, rewriting based on the above rule would result in this SQL query where the custom aggregate \textit{MinCostSuppAgg} is defined as given in \reffig{fig:udf-agg}. Observe that out of the four parameters to the aggregate function, two are attributes of the cursor query Q, one is a local variable in the UDF, and the other is a parameter to the UDF.

\reffig{fig:jdbc-trans} shows the Java method from \reffig{fig:jdbc-example} rewritten using Aggify. Out of the 2 parameters to the aggregate function, one is an attribute from the underlying query, and the other is a local variable. The loop is replaced with a method that advances the \textit{ResultSet} to the first (and only) row in the result, and an attribute accessor method invocation (\textit{getDouble()} in this case) is placed with an assignment to each of the \textit{live} variables (\textit{cumulativeROI} in this case).

\subsection{Order enforcement}
The query $Q$ over which a cursor is defined may be arbitrarily complex. If $Q$ does not have an ORDER BY clause, the DBMS gives no guarantee about the order in which the rows are iterated over. \refeq{eq:rule1a} is in accordance with this, because the DBMS gives no guarantee about the order in which the custom aggregate is invoked as well. Hence the above query rewrite suffices in this case.

However, the presence of ORDER BY in the cursor query $Q$ implies that the loop body $\Delta$ is invoked in a specific order determined by the sort attributes of $Q$. In this case, the above rewriting is not sufficient as it does not preserve the ordering and may lead to wrong results. 
Therefore, Simhadri et. al~\cite{uudf} mention that either there should be no ORDER BY clause in the cursor query, or the database system should allow order enforcement while invoking custom aggregates.
To address this, we now propose a variation of the above rewrite rule that can be used to enforce the necessary order. 

Let $Q_s$ represent a query with an ORDER BY clause where the subscript $s$ denotes the sort attributes. Let $Q$ represent the query $Q_s$ without the ORDER BY clause.
For a cursor loop  ($Q_s, \Delta$), the rewrite rule can be stated as follows:
\begin{equation} \label{eq:rule1b}
    \text{Loop}(Q_s,\Delta) \implies \mathcal{G}_{\text{StreamAgg}_{\Delta}(P_{\textit{accum}}) \text{ as aggVal}}(Sort_s(Q))
\end{equation}

This rule enforces the following two conditions. (i) It enforces the sort operation to be performed before the aggregate is invoked, and (ii) it enforces the \textit{Streaming Aggregate} physical operator to implement the custom aggregate. These two conditions ensure that the order specified in the cursor loop is respected.

\subsection{Discussion}
Once the query is rewritten as described above, Aggify replaces the loop with an invocation to the rewritten query as shown in Figures \ref{fig:udf-trans} and \ref{fig:jdbc-trans}. The return value of the aggregate is assigned to corresponding local variables, which enables subsequent lines of code to remain unmodified. %These modifications are straightforward, and we omit details.
From Figures \ref{fig:udf-trans} and \ref{fig:jdbc-trans}, we can make the following observations.

\begin{itemize}
    \item The cursor query Q remains unchanged, and is now the subquery that appears in the FROM clause.
    \item The transformation is fairly non-intrusive. Apart from the removal of the loop, the rest of the lines of code remain identical, except for a few minor modifications.
    \item This transformation may render some variables as \textit{dead}. Declarations of such variables can be then removed, thereby further simplifying the code~\cite{DFABOOK}. For instance, the variables \textit{@pCost} and \textit{@sName} in \reffig{fig:udf-example} are no longer required, and are removed in \reffig{fig:udf-trans}.
\end{itemize}

The transformed program that is output by Aggify offers the following benefits. It avoids materialization of the cursor query results and instead, the entire loop is now a single pipelined query execution.
In the context of loops in applications that run outside the DBMS (\reffig{fig:jdbc-trans}), this rewrite significantly reduces the amount of data transferred between the DBMS and the client. Further, the entire loop computation which ran on the client now runs inside the DBMS, closer to data. Finally, all these benefits are achieved without having to perform intrusive changes to existing source code, as we observed above. 
% Allows cursors on CCI

\subsection{Aggify Algorithm}

\begin{algorithm} 
    %\color{blue}
    \begin{algorithmic}%[1] 
    \caption{Aggify($G$, $Q$, $\Delta$)} \label{algorithm:aggify}
    
        \Require $G$: CFG of the program augmented with data dependence edges; \\ $Q$: Cursor query; \\ $\Delta$: Subgraph of $G$ for the loop body; \newline
        
        %\Procedure{TransformLoopToAgg}{G, Q, $\Delta$}
        \State $A (L, R, UD, DU)$ \textleftarrow \hphantom{.} Perform DataFlow Analysis on $G$;
        \State \hskip2.0em $L$ \textleftarrow \hphantom{.} Liveness information;
        \State \hskip2.0em $R$ \textleftarrow \hphantom{.} Reachable Definitions;
        \State \hskip2.0em $UD$ \textleftarrow \hphantom{.} Use-Def Chain;
        \State \hskip2.0em $DU$ \textleftarrow \hphantom{.} Def-Use Chain; \newline
        
        \State $V_{\Delta}$ \textleftarrow \hphantom{.} \{Variables referenced in $\Delta$\};
        \State $V_{fetch}$ \textleftarrow  \hphantom{.} \{Vars. assigned in the FETCH statement\};

        \State $V_{field}$ \textleftarrow  \hphantom{.} \{Compute using Equation~\refeq{eq:fields}\};
        \State $P_{accum}$ \textleftarrow  \hphantom{.} \{Compute using Equation~\refeq{eq:paccum}\};
        \State $V_{init}$ \textleftarrow  \hphantom{.} \{Compute using Equation~\refeq{eq:vinit}\};
        \State $V_{term}$ \textleftarrow  \hphantom{.} \{Fields that are live at loop end\};\newline
        
        \State $Agg_{\Delta}$ \textleftarrow  \hphantom{.} Construct aggregate class using template in \\ \reffig{fig:agg-template} and above information;
        \State Register $Agg_{\Delta}$ with the database;\newline
        
    // Replace loop in $G$ with rewritten query
        \If{(Q contains ORDER BY clause)}
            \State $s$ \textleftarrow \hphantom{.} \{ORDER BY attributes\}
            \State Rewrite loop using Equation \refeq{eq:rule1b};     
        \Else
            \State Rewrite loop using Equation \refeq{eq:rule1a};
        \EndIf
    \end{algorithmic} 
\end{algorithm}

The entire algorithm illustrated in sections \refsec{sec:trans} and \refsec{sec:rewrite} is formally presented in Algorithm \ref{algorithm:aggify}. The algorithm accepts $G$, the CFG of the program augmented with data dependence edges, $Q$, the cursor query, and $\Delta$, the subgraph of $G$ corresponding to the loop body. Algorithm \ref{algorithm:aggify} is invoked after all the necessary preconditions in \refsec{subsec:support} are satisfied.

Initially, we perform DataFlow Analyses on $G$ as described in \refsec{sec:background}. The results of these analyses are captured as $A(L, R, UD, DU)$ which consists of Liveness, Reachable definitions, Use-Def and Def-Use chains respectively. Then, these results are used to compute the necessary components for the aggregate definition, namely $V_{\Delta}$, $V_{init}$, $V_{fetch}$, $V_{field}$, $V_{term}$, $P_{accum}$. 
Once all the necessary information is computed, the aggregate definition is constructed using the template in \reffig{fig:agg-template}, and this aggregate (called $Agg_{\Delta}$) is registered with the database engine.

Finally, we rewrite the entire loop with a query that invokes $Agg_{\Delta}$. The rewrite rule is chosen based on whether the cursor query $Q$ has an ORDER BY clause, as described in \refsec{sec:rewrite}.

\subsubsection{Nested cursor loops:} \label{subsubsec:nested}

Cursors can be nested, and our algorithm can handle such cases as well. This can be achieved by first running Algorithm \ref{algorithm:aggify} on the inner cursor loop and transforming it into a SQL query. Subsequently, we can run Algorithm \ref{algorithm:aggify} on the outer loop. An example is provided (L8-W2) in customer workload experiments in~\cite{AGGIFYWL}.

\section{Preserving Semantics} \label{sec:preserve}
We now reason about the correctness of the transformation performed by Aggify, and describe how  the semantics of the cursor loop are preserved. 

Let CL($Q, \Delta$) be a cursor loop, and let $Q'$ be the rewritten query that invokes the custom aggregate. 
The program state comprises of values for all \textit{live} variables at a particular program point. Let $P_0$ denote the program state at the beginning of the loop and $P_n$ denote the program state at the end of the loop, where $n = |Q|$. 
To ensure correctness, we must show that if the execution of the cursor loop on $P_0$ results in $P_n$, then the execution of $Q'$ on $P_0$ also results in $P_n$. We only consider program state and not the database state in this discussion, as our transformation only applies to loops that do not modify the database state. 

Every iteration of the loop can be modeled as a function that transforms the intermediate program state. Formally,
\begin{equation*}
     P_{i} = f(P_{i-1}, T_i) 
\end{equation*}
where $i$ ranges from 1 to $n$. In fact the function $f$ would be comprised of the operations in the loop body $\Delta$.

It is now straightforward to see that the \textit{Accumulate()} method of the custom aggregate constructed by Aggify exactly mimics this behavior. This is because (a) the statements in the loop body $\Delta$ are directly placed in the \textit{Accumulate()} method,  (b) the \textit{Accumulate()} is called for every tuple in $Q$, and (c) the rule in \refeq{eq:rule1b} ensures that the order of invocation of \textit{Accumulate()} is identical to that of the loop when necessary. The fields of the aggregate class\footnote{Note that here $V_F = P_0$. In other words, we consider all variables that are live at the beginning of the loop (i.e. $P_0$) as fields of the aggregate. This is a conservative but correct definition as given in \refsec{subsec:fields}.} and their initialization ensure the inter-iteration program states are maintained. 
%Observe that based on our definition of program state, $P_n \subseteq P_0$. 
From our definition of $V_{term}$ in \refsec{subsec:terminate}, it follows that $P_n = V_{term}$. Therefore the output of the custom aggregate is identical to the program state at the end of the cursor loop.\qedwhite

\section{Enhancements} \label{sec:enhance}
We now present enhancements to Aggify that further improve the rewritten code and broaden its applicability.

\subsection{Simplifying the Custom Aggregate}
During construction of the custom aggregate, we placed the entire loop body in the \textit{Accumulate()} method (\refsec{subsec:accum}). Note that while this is correct, it might not be the most efficient. This is because operations inside \textit{Accumulate()} are not visible to the query optimizer. So, it would be preferable to minimize the number of operations inside \textit{Accumulate()} and expose more operations to the query optimizer.

To this end, we use \textit{loop invariant code motion}~\cite{KENNEDYBOOK}, a standard compiler optimization technique to move loop-invariant operations outside the cursor loop. In fact, we can additionally move loop-variant expressions outside the loop body into the cursor query $Q$, if the expression does not involve any variable that is written to in the loop body. We call this transformation as \textit{acyclic code motion}. This can be highly beneficial, especially if relational operations can be pulled out of the loop and merged with the cursor query.
Simhadri et. al~\cite{uudf} propose that the statements that precede the start of a data dependence cycle~\cite{KENNEDYBOOK} can be safely moved out of the loop. We broaden this further and show that even within statements that are  part of a data dependence cycle, expressions can be pulled out. %Similar techniques have been proposed by~\cite{GUR08}.

A simple example can be illustrated using the loop in \reffig{fig:udf-example}. The boolean expression \textit{(@pCost > @lb)} in the loop qualifies for acyclic code motion because there are no assignments to both \textit{@lb} and \textit{@pCost} in the loop. Therefore this expression can be moved out of the loop and merged with the cursor query as follows:

\begin{verbatim}
SELECT ps_supplycost, s_name, 
            (ps_supplycost > @lb) as boolVal
  FROM PARTSUPP, SUPPLIER 
 WHERE ps_partkey= @pkey 
   AND ps_suppkey= s_suppkey 
\end{verbatim}

This kind of rewrite now brings such expressions into the SQL query and simplifies the logic in the custom aggregate. The query optimizer may then use parallel evaluation or batched execution to efficiently evaluate these expressions. The implementation of acyclic code motion inside Aggify is currently under progress.

\subsection{Optimizing Iterative FOR Loops} \label{subsec:for}
Although the focus of Aggify has been to optimize loops over query results, the technique can in principle be extended to more general FOR loops with a fixed iteration space. 
A FOR loop is a control structure used to write a loop that needs to execute a specific number of times. Such loops are extremely common, and typically have the following structure: 

\begin{verbatim}
  FOR (init; condition; increment) { statement(s);}
\end{verbatim}

Such loops can in fact be rewritten as cursor loops by expressing the iteration space as a relation. For instance, consider this loop

\begin{verbatim}
  FOR (i = 0; i <= 100; i++) { statement(s);}
\end{verbatim}

The iteration space of this loop can be written as a SQL query using either recursive CTEs, or vendor specific constructs (such as DUAL in Oracle). The above loop written using a recursive CTE is given below.

\begin{verbatim}
    with CTE as (  
     select 0 as i  
     union all  
     select i + 1  from CTE where i <= 100  
    )  select * from CTE 
\end{verbatim}

Now, we can define a cursor on the above query with the same loop body. This is now a standard cursor loop and hence, Aggify can be used to optimize it.
Rewriting FOR loops as recursive CTEs can be achieved by extracting the \textit{init}, \textit{condition} and \textit{increment} expressions from the FOR loop, and placing them in the CTE template given above. Note that these expressions may be arbitrarily complex, involving program variables, and the values need not be statically determinable. We omit details of this transformation.% due to lack of space. 

% For example, consider the program snippet in \reffig{fig:iter-loop}. This program trains a simple linear regression model using stochastic gradient descent. It uses 

% for (i = 0; i < numEpochs; i++)
% {
%   declare cursor over dataset;
%   while (@@Fetch_status)
%     compute coeffs

% }

\subsection{Extending existing techniques} \label{subsec:integ}
We now show how Aggify can seamlessly integrate with existing optimization techniques, both in the case of applications that run outside the DBMS and UDFs that run within.

There have been recent efforts to optimize database applications using  techniques from programming languages~\cite{Cheung13, Emani2016, cobra18, Radoi14}. In fact, \cite{Emani2016} mentions that loops over query results could be converted to user-defined aggregates, but do not describe a technique for the same. Aggify can be used as a preprocessing step which can replace loops with equivalent queries which invoke a custom aggregate. Then, the techniques of \cite{Emani2016} can be used to further optimize the program.

The Froid framework~\cite{FroidVldb}, which was also based on the work of Simhadri et. al~\cite{uudf} showed how to transform UDFs into sub-queries that could then be optimized. However, Froid cannot optimize UDFs with loops.
Building upon this idea, Duta et. al~\cite{Grust2019} described a technique to transform arbitrary loops into recursive CTEs. While this avoids function call overheads and context switching, it can sometimes perform worse due to the limited optimizations that currently exist for recursive CTEs. For the specific case of cursor loops, Aggify avoids creating recursive CTEs. 

These ideas can be used together in the following manner: (i) If the function has a cursor loop, use Aggify to eliminate it. (ii) If the function has a FOR loop and the necessary expressions can be extracted from it, use the technique described in \refsec{subsec:for} along with Aggify to eliminate the loop. (iii) If the function has an arbitrary loop with a dynamic iteration space, use the technique of Duta et. al~\cite{Grust2019}. After applying (i), (ii), or (iii), Froid can be used to optimize the resulting loop-free UDF.
\section{Design and Implementation} \label{sec:impl}
The techniques described in this paper can be implemented either inside a DBMS or as an external tool. 
% We have currently implemented Aggify as a standalone tool that performs the transformations as described in this paper. Programs (UDFs or stored procedures) with cursor loops are fed as input. Then, Aggify produces two outputs: a custom aggregate definition, and a rewritten program in which the cursor loop is replaced by an SQL query that invokes this custom aggregate. 
We have currently implemented a prototype of Aggify in Microsoft SQL Server. 
Aggify currently supports cursor loops that are written in Transact-SQL~\cite{TSQL}, and constructs a user-defined aggregate in C\#~\cite{CLRU}. Note that translating from T-SQL into C\# might lead to loss of precision and sometimes different results due to difference in data type semantics. We are currently working on a better approach which is to natively implement this inside the database engine~\cite{AGGIFYTR}.

% Since Aggify is currently an external tool, we have implemented custom aggregates as user-defined aggregates. 
Implementing Aggify inside a DBMS allows the construction of more efficient implementations of custom-aggregates that can be baked into the DBMS itself. 
Also, observe that the rule in \refeq{eq:rule1b} that enforces streaming aggregate operator for the custom aggregate has to be part of the query optimizer. In fact, apart from this rule, there is no other change required to be made to the query optimizer. Every other part of Aggify can be implemented outside the query optimizer. However, we note that since Microsoft SQL Server only supports the Streaming Aggregate operator for user-defined aggregates, we did not have to implement \refeq{eq:rule1b}. 

Froid~\cite{FroidVldb} is available as a feature called \textit{Scalar UDF Inlining}~\cite{UDFInlining} in Microsoft SQL Server 2019. As mentioned in \refsec{subsec:integ}, Aggify integrates seamlessly with Froid, thereby extending the scope of Froid to also handle UDFs with cursor loops. Aggify is first used to replace cursor loops with equivalent SQL queries with a custom aggregate; this is then followed by Froid which can now inline the UDF.

%not cost-based

%Enable cursors on CCI

%CLR semantics mismatch, 

% CLR perf gains

\section{Evaluation} \label{sec:eval}
We now present some results of our evaluation of Aggify.
%\subsection{Experimental Setup} \label{subsec:setup}
Our experimental setup is as follows. Microsoft SQL Server 2019 was run on Windows 10 Enterprise. The machine was equipped with an Intel Quad Core i7, 3.6 GHz, 64 GB RAM, and SSD-backed storage. The SQL cursor loops in UDFs/stored procedures was run directly on this machine, and hence did not involve any network usage. The database-backed applications were run from another desktop machine with similar configuration and 32GB RAM, connected to the DBMS over LAN.

\subsection{Workloads} \label{subsec:setup}
We have evaluated Aggify on many workloads on several data sizes and configurations. We show results based on 3 real workloads, an open benchmark based on TPC-H queries, and several Java programs including an open benchmark. All these queries, UDFs and cursor loops are made available~\cite{AGGIFYWL}.   

\mysubsection{TPC-H Cursor Loop workload}
To evaluate Aggify on an open benchmark that mimics real workloads, we implemented the specifications of a few TPC-H queries using cursor loops. Not all TPC-H queries are amenable to be written using cursor loops, so we have chosen a logically meaningful subset. While this is a synthetic benchmark, it illustrates common scenarios found in real workloads. We report results on the 10GB scale factor. The database had indexes on L\_ORDERKEY and L\_SUPPKEY columns of LINEITEM, O\_CUSTKEY column of ORDERS, and PS\_PARTKEY column of PARTSUPP. For this workload we show a breakdown of results for Aggify, and Aggify+ (Froid applied after Aggify). 

\mysubsection{Real workloads}
We have considered 3 real workloads (proprietary) for our experiments. 
While we have access to the queries and UDFs/stored procedures of these workloads, we did not have access to the data. As a result, we have synthetically generated datasets that suit these workloads. We have also manually extracted required program fragments from these workloads so that we can use them as inputs to Aggify. %These are proprietary and hence we will not be able to show the source code of UDFs or queries for those.
Workload W1 is a CRM application, W2 is a configuration management tool, and W3 is a backend application for a transportation services company. 
Note that we have not combined Aggify with Froid in these workloads, as we did not have access to the queries that invoke these UDFs.

\mysubsection{Java workload}
As mentioned in \refsec{sec:impl}, the implementation of Aggify for Java is ongoing. For the purpose of detailed performance evaluation, we have considered the RUBiS benchmark~\cite{Rubis} and two other examples and manually transformed them using Algorithm~\ref{algorithm:aggify}. One of them is a Java implementation of the minimum cost supplier functionality similar to the example in \reffig{fig:udf-example}. The other is a variant of the example in \reffig{fig:jdbc-example} with 50 columns. These programs are also available in \cite{AGGIFYWL}. Note that Froid is applicable only to T-SQL, so the Java experiments do not use Froid.

\subsection{Applicability of Aggify} \label{subsec:eval-applicability}
\begin{table}
	\small
	%\color{blue}
	\centering	
	\caption{Applicability of Aggify}	
	\vspace{-1em}
	\begin{tabular}{ |c|c|c|c| } 
		\hline
		\textbf{Workload} &  \textbf{RUBiS} & \textbf{RUBBoS} & \textbf{Adempiere}\\
		\hline 
		\text{Total \# of while loops} & 16 & 41 & 127 \\
		\hline
		\text{\# of cursor loops} & 14 (87.5\%)  & 14 (34.14\%) & 109 (85.8\%) \\
		\hline
		\text{Aggify-able} & 14 & 14 & >80 \\
		\hline
	\end{tabular}
	\label{tab:applicability} 
	\vspace{-1em}
\end{table}

We have analyzed several real world workloads and opensource benchmark applications to measure (a) the usage of cursor loops, and (b) the applicability of Aggify on such loops. 
We considered about 5720 databases in Azure SQL Database~\cite{SQLDB} that make use of UDFs in their workload. Across these databases, we came across more than 77,294 cursors being declared inside UDFs\footnote{This analysis was done using scripts that analyze the database metadata and extract the necessary information. We did not have access to manually examine the source code of these propreitary UDFs.}. As explained in \refsec{subsec:applicability}, Aggify can be used to rewrite all these cursor loops. This demonstrates both the wide usage of cursor loops in real workloads, and the broad applicability of Aggify.

Next, we manually analyzed 3 opensource Java applications -- the RUBiS Benchmark~\cite{Rubis}, RUBBos Benchmark~\cite{Rubbos}, and the popular Adempiere~\cite{Adempiere} CRM application. Table~\ref{tab:applicability} shows the results of this analysis. 87.5\% of the loops in the RUBiS benchmark were cursor loops. In RUBBoS, 34\% of the loops were cursor loops. We looked at a subset of files in Adempiere (\~25 files) where more than 85\% of the loops encountered were cursor loops. Table ~\ref{tab:Adempiere} lists these 25 files and the number of aggifyable cursor loops in each. Interestingly all the cursor loops in RUBiS and RuBBoS, and more than 80 loops in Adempiere satisfied the preconditions for Aggify. This shows the use of cursors as well as the applicability of Aggify.

\begin{table}
	\small
	\centering	
	\caption{Analysis of Adempiere Application}		
	\vspace{-1em}
	\begin{tabular}{ |c|c|c|c| } 
	    \hline
		\textbf{File Name} & \textbf{\#While} & \textbf{\#Cursor} & \textbf{\#Aggifyable}\\
		\hline		
		SmjReportLogic & 1 & 1 & 1\\
		\hline
		DbDifference & 36 & 30 & 29\\
		\hline
		WebInfo & 17 & 17 & 17\\
		\hline
		PrintBOM & 7 & 6 & 0\\
		\hline
		MRequestProcessor & 3 & 3 & 3\\
		\hline
		TranslationController & 3 & 3 & 3\\
		\hline
		JDBCInfo & 3 & 3 & 0\\
		\hline
		MIMPProcessor & 3 & 3 & 3\\
		\hline
		MWebServiceType & 3 & 3 & 3\\
		\hline
		SequenceCheck & 3 & 3 & 0\\
		\hline
		ScheduleUtil & 7 & 3 & 0\\
		\hline
		Invoice & 4 & 4 & 4\\
		\hline
		Login & 8 & 8 & 0\\
		\hline
		MSearchDefinition & 2 & 2 & 2\\
		\hline
		MWebService & 2 & 1 & 1\\
		\hline
		MLocator & 2 & 2 & 0\\
		\hline
		DocumentTypeVerify & 2 & 2 & 0\\
		\hline
		Payment & 1 & 1 & 1\\
		\hline
		ImpFormat & 5 & 1 & 1\\
		\hline
		Compiere & 2 & 0 & 0\\
		\hline
		MStorage & 5 & 5 & 5\\
		\hline
		FixPaymentCashLine & 2 & 2 & 1\\
		\hline
		CreateFrom & 2 & 2 & 2\\
		\hline
		MTreeNodeCMS & 1 & 1 & 1\\
		\hline
		MTreeNodeCMC & 1 & 1 & 1\\
		\hline
% 		MTreeNodeCMC & 1 & 1 & 1\\
% 		\hline
		
	\end{tabular}
	\label{tab:Adempiere} 
%	\vspace{-0.7em}
\end{table}

\subsection{Performance improvements} \label{subsec:perf}
\begin{figure*}[t!]
	\minipage{0.65\columnwidth}
	\includegraphics[width=\linewidth]{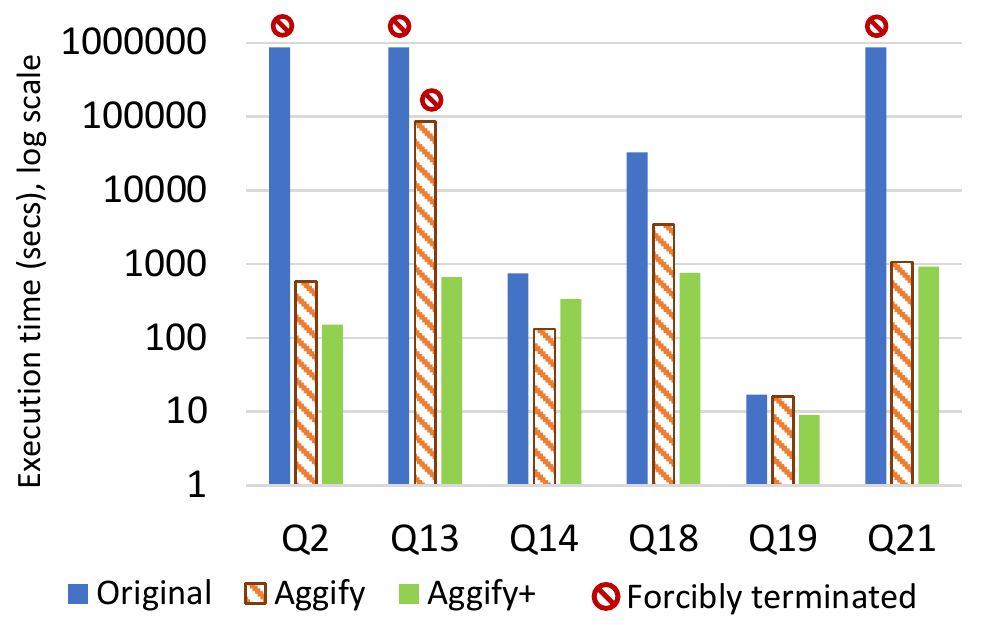}
	\caption*{(a) TPC-H cursor loop workload}\label{fig:tpch10g-warm}
	\endminipage
	\hspace{2mm}
	\minipage{0.65\columnwidth}
	\includegraphics[width=\linewidth]{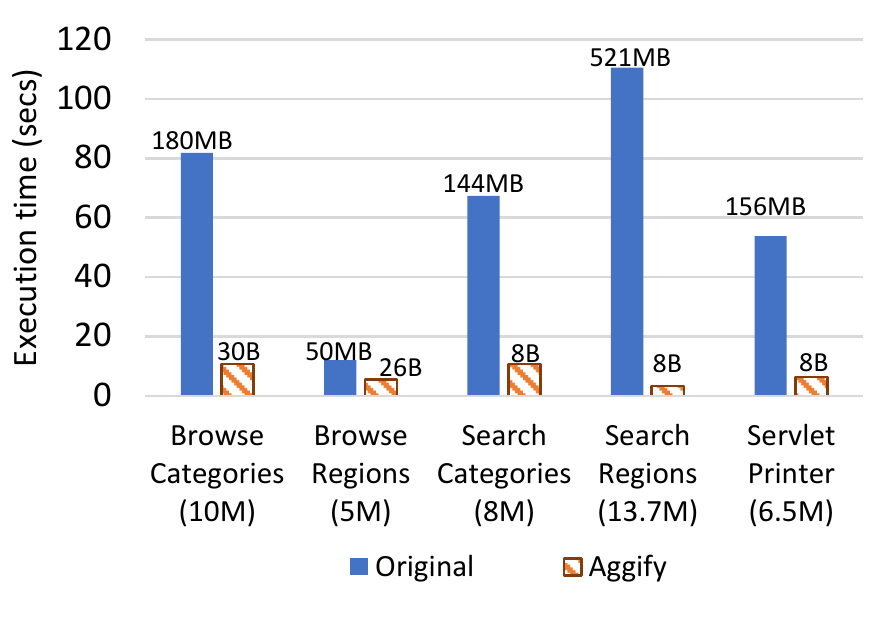}
	\caption*{(b) Java workload}\label{fig:java-perf}
	\endminipage
	\hspace{2mm}
	\minipage{0.65\columnwidth}
	\includegraphics[width=\linewidth]{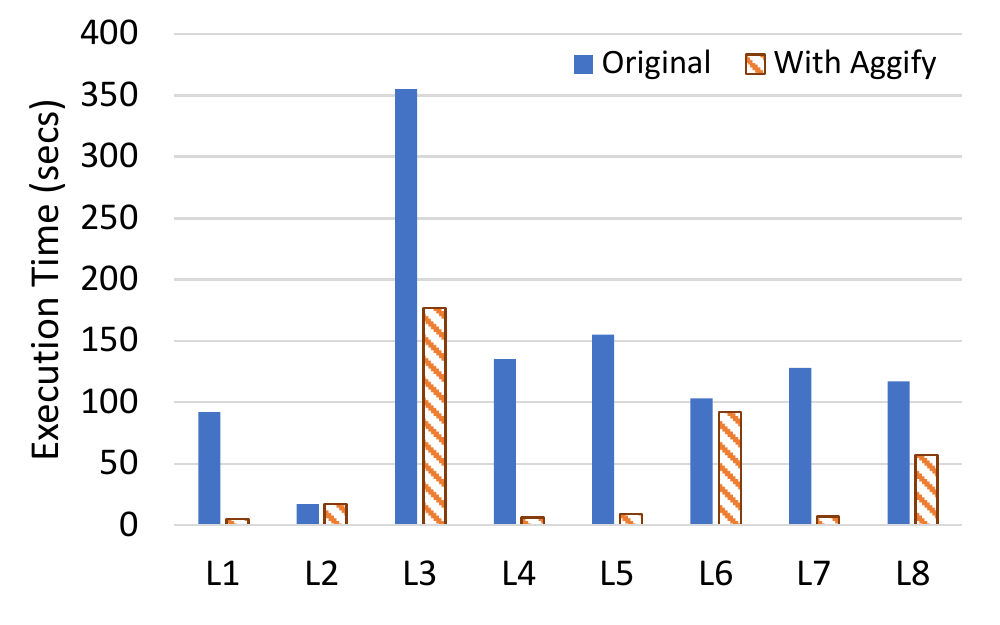}
    \caption*{(c) Customer workloads}\label{fig:cust-all-perf}
	\endminipage
	\vspace{-3mm}
	\caption{Performance improvements across multiple workloads} \label{fig:perf}
\end{figure*}

% \onecolfigure
% {figs/eval/tpch10g-warm}
% {Performance with and without Aggify for the TPC-H cursor loop workload.}
% {fig:tpch10g-warm}

% \onecolfigure
% {figs/eval/cust-all-perf}
% {Performance with and without Aggify for real workloads (selected from W1,W2,W3).}
% {fig:cust-all-perf}

We now show the results of our experiments to determine the overall performance gains achieved due to Aggify. 

\subsubsection{TPC-H workload}
First, we consider the TPC-H cursor loop workload and show the results on a 10 GB database with warm buffer pool. Similar trends have been observed with 100GB as well, with both warm and cold buffer pool configurations. 
\reffig{fig:perf}(a) shows the results for 6 queries from the workload~\cite{AGGIFYWL}. The solid column (in blue) represents the original program. The striped column (orange) represents results of applying Aggify. The green column (indicated as `Aggify+' shows the results of applying Froid's technique after Aggify enables it, as described in \refsec{subsec:integ}.

On the x-axis, we indicate the query number, and on the y-axis, we show execution time in seconds, in log scale. Observe that for queries Q2, Q13 and Q21, we have a $\boldsymbol{\oslash}$ symbol above the column corresponding to the original query with the loop, and for Q13, we have that symbol even for the Aggify column. This means that we had to forcibly terminate these queries as they were running for a very long time (\textgreater 10  days for Q2,  \textgreater 22 days for Q13 and \textgreater 9 hours for Q21). We observe that Q2, Q14, Q18 and Q21 offer at least an order of magnitude improvement purely due to Aggify alone.
When Aggify is combined with Froid, we see further improvements in Q2, Q13, Q18 and Q19. Q13 results in a huge improvement of 3 orders of magnitude due to the combination. Note that without Aggify, Froid will not be able to rewrite these queries at all. Q21 does not lead to any additional gains from Froid, while Q14 slows down slightly due to Froid.

% \renewcommand{\arraystretch}{1.2}
% \begin{table*}
% 	\small
% 	\centering	
% 	\caption{Benefits of Aggify on customer workloads.}	
% 	\vspace{-1em}
% 	\begin{tabular}{ |c|c|c|c|c|c|c| } 
% 		\hline
% 		\textbf{Loop} & \textbf{Workload} & \textbf{Original} & \textbf{w/Aggify} & \textbf{Improvement factor} & \textbf{Iteration count} & \textbf{Comments} \\
% 		\hline		
% 		L1 & W2 & 92s & 5s & 18x & 5M & No table variable inserts\\
% 		\hline
% 		L2 & W1 & 17s & 17s & 1x & 10K & Large number of table inserts\\
% 		\hline
% 		L3 & W1 & 355s & 177s & 2x & 9K & Includes table inserts\\
% 		\hline
% 		L4 & W2 & 135s & 6s & 22x & 7M & No table variable inserts\\
% 		\hline
% 		L5 & W2 & 155s & 9s & 17x & 7M & No table variable inserts\\
% 		\hline
% 		L6 & W2 & 103s & 92s & 1.1x & 40K & One table insert statement per iteration\\
% 		\hline
% 		L7 & W3 & 128s & 7s & 18x & 7M & No table variable inserts\\
% 		\hline
% 		L8 & W2 & 117s & 57s & 2x & 3*2M & Nested loop (outer * inner)\\
% 		\hline
% 	\end{tabular}
% 	\label{tab:cust} 
% %	\vspace{-0.7em}
% \end{table*}
% \renewcommand{\arraystretch}{1.0}

\begin{table}
	\small
	\centering	
	\caption{Loops from customer workloads.}		
	\vspace{-1em}
	\begin{tabular}{ |c|c|c| } 
	    \hline
		\textbf{Loop} & \textbf{\#Iterations} & \textbf{Comments}\\
		\hline		
		L1(W2) & 5M & No table variable inserts\\
		\hline
		L2(W1) & 10K & Large number of table inserts\\
		\hline
		L3(W1) & 9K & Includes table inserts\\
		\hline
		L4(W2) & 7M & No table variable inserts\\
		\hline
		L5(W2) & 7M & No table variable inserts\\
		\hline
		L6(W2) & 40K & Includes table inserts\\
		\hline
		L7(W3) & 7M & No table variable inserts\\
		\hline
		L8(W2) & 3*2M & Nested loop (outer * inner)\\
		\hline
	\end{tabular}
	\label{tab:cust} 
%	\vspace{-0.7em}
\end{table}

\subsubsection{Java workload}
\reffig{fig:perf}(b) shows the results of running Aggify on 5 loops from the RUBis~\cite{Rubis} benchmark. The x-axis indicates the 5 scenarios along with the number of iterations of the loop (given in parenthesis), and y-axis shows the execution time. As before, the blue column indicates the original program, and the dashed orange column indicates the results with Aggify. Note that Froid is applicable only to T-SQL and not to Java.
%We have also shown the amount of data transferred in each case, above the columns. All the loops used in the experiment are made available~\cite{AGGIFYWL}. 
We observe that Aggify improves performance for all these scenarios. Here the beneifts due to Aggify stem mainly from the huge reduction in data transfer between the database server and the client Java application. 

\subsubsection{Real workloads}
Now, we consider some loops that we encountered in customer workloads W1, W2 and W3 and run them with and without Aggify. \reffig{fig:perf}(c) shows the results of using Aggify on 8 of these loops. The y-axis shows execution time in seconds for loops L1-L8. Table~\ref{tab:cust} shows additional information about the loops chosen including the iteration count for each loop, and whether the loop performed any inserts on table variables. 
We observe improvements in most cases, ranging from 2x to 22x. Note that loop L8 in \reffig{fig:perf}(c) is a nested cursor loop that gives more than 2x gains. Loops L2 and L6 iterate over a relatively small number of tuples compared to the others. This is one cause for the small or no performance gains.
Also, these two loops included many statements that inserted values into temporary tables or table variables. For such statements, in our implementation of Aggify, we have to make a connection to the database explicitly in order to insert these tuples. That adds additional overhead, which could be avoided if the aggregate is implemented natively inside the DBMS (\refsec{sec:impl}). 

%\mysubsection{Impact on concurrency and throughput}

\subsection{Resource Savings}
% \onecolfigure
% {figs/eval/tpch10g-warmReads}
% {Comparison of logical reads for the TPC-H cursor loop benchmark.}
% {fig:tpch10g-warmReads}
\newcolumntype{g}{>{\columncolor{columbiablue}}c}
\begin{table}
	\small
	\color{black}
%	\centering	
	\caption{Comparison of logical reads for the TPC-H cursor loop benchmark.}	
	\vspace{-1em}
	\begin{tabular}{ |c|c|c|c|c|c|} 
		\hline
		\multirow{2}{*}{\textbf{Qry}} & \multirow{2}{*}{\textbf{Original}}&
		\multirow{2}{*}{\textbf{Aggify}}&
		\multirow{2}{*}{\textbf{Aggify+}}&
        \multicolumn{2}{c|}{\textbf{Savings w.r.t Original}}\\
        \cline{5-6}
        & & & & \textbf{Aggify} & \textbf{Aggify+} \\
		%& \multicolumn{1}{|p{1.4cm}|}{\centering \textbf{Aggify+} \\ \textbf{w.r.t Orig.}}\\
		\hline		
		Q2 & $\boldsymbol{\oslash}$ & 38.1M
                            & 54.4M & NA & NA\\
		\hline
		Q13 & $\boldsymbol{\oslash}$ & $\boldsymbol{\oslash}$ & 0.26M & NA & NA\\
		\hline
		Q14 & 553M & 11.7M & 231M & 541M & 322M \\
		\hline
		Q18 & 405M &  120M & 293M & 285M & 113M\\
		\hline
		Q19 & 1.11M & 1.11M & 1.11M & 4528 & 4611\\
		\hline
		Q21 & $\boldsymbol{\oslash}$ & 464M & 616M & NA & NA\\
		\hline
	\end{tabular}
	\label{tab:resource} 
	\vspace{-0.7em}
\end{table}

In addition to performance gains, Aggify also reduces resource consumption, primarily disk IO. This is because cursors end up materializing query results to disk, and then reading from the disk during iteration, whereas the entire loop runs in a pipelined manner with Aggify. 
To illustrate this, we measured the logical reads incurred on our workloads. Table~\ref{tab:resource} shows these numbers (in millions) for 6 queries from the TPC-H cursor loop benchmark. For the original programs, we have the numbers for 3 queries as we had to forcibly terminate the others as mentioned in \refsec{subsec:perf}. 

From the results, we can see that Aggify significantly brings down the required number of reads. Q14 and Q18 show huge reductions (58\% and 27\% respectively). Table~\ref{tab:resource} shows the breakup of logical reads with \textit{Aggify} alone (column 3), and \textit{Aggify+} (column 4) which denotes Aggify with Froid. Columns 5 and 6 denote the savings in logical reads due to Aggify and Aggify+. Interestingly, we observe that using Froid with aggify results in more logical reads, but improves execution time.

\subsection{Scalability} \label{subsec:scale}

\begin{figure*}[t!]
	\minipage{0.65\columnwidth}
	\includegraphics[width=\linewidth]{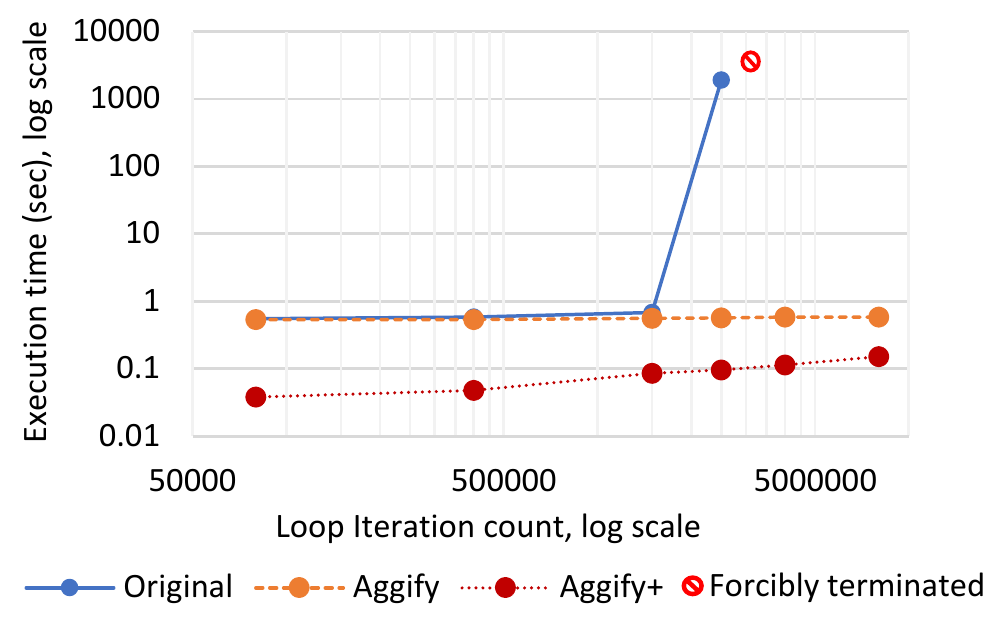}
	\caption*{(a) MinCostSupplier (SQL).}\label{fig:tpch10g-varycard}
	\endminipage
	\hspace{2mm}
	\minipage{0.65\columnwidth}
	\includegraphics[width=\linewidth]{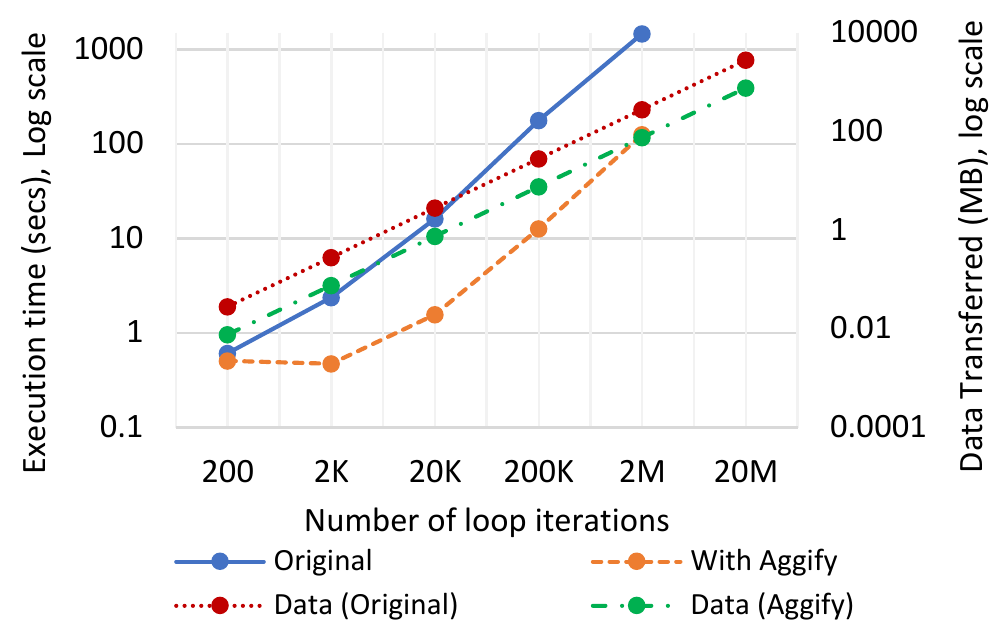}
	\caption*{(b) MinCostSupplier (Java)}\label{fig:mcs-java-varycard10g}
	\endminipage
	\hspace{2mm}
	\minipage{0.65\columnwidth}
	\includegraphics[width=\linewidth]{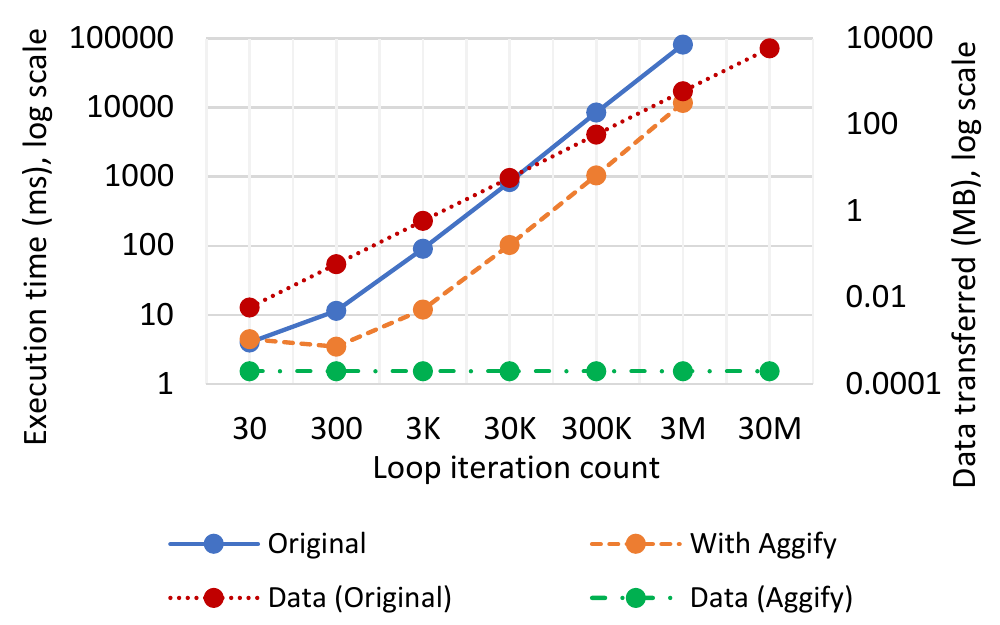}
    \caption*{(c) CumulativeROI (Java).}\label{fig:roi-vary-card}
	\endminipage
	\vspace{-3mm}
	\caption{Scalability across multiple workloads (varying loop iteration counts).}
	\label{fig:scale}
\end{figure*}

% \begin{figure*}[t!]
% 	\minipage{0.65\columnwidth}
% 	\includegraphics[width=\linewidth]{figs/eval/scalability/TMW-scalability}
% 	\caption*{Scalability for loop L1 (Workload W2)}\label{fig:L1-W2_card}
% 	\endminipage
% 	\hspace{2mm}
% 	\minipage{0.65\columnwidth}
% 	\includegraphics[width=\linewidth]{figs/eval/scalability/rubis_scalability_log.pdf}
% 	\caption*{Scalability for Browse Categories (Rubis Workload)}\label{fig:Rubis-Scalability}
% 	\endminipage

% 	\vspace{-3mm}
% 	\caption{Scalability across Open source and customer workloads (varying loop iteration ).}
% 	\label{fig:scale}
% \end{figure*}

\onecolfigure
{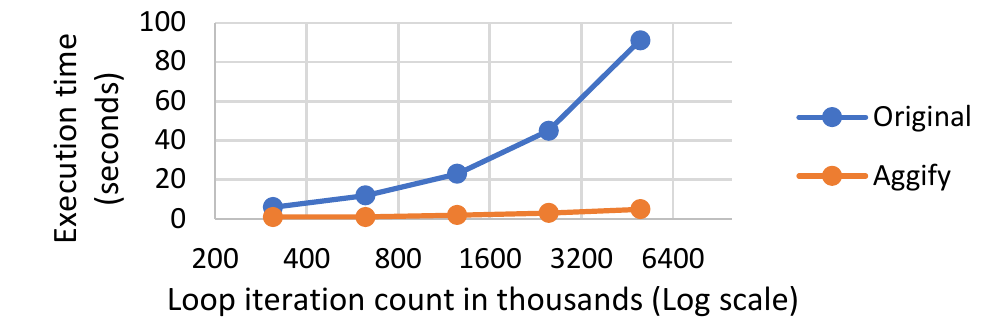}
{Scalability for loop L1 (workload W2)}
{fig:tmw-scale}

\onecolfigure
{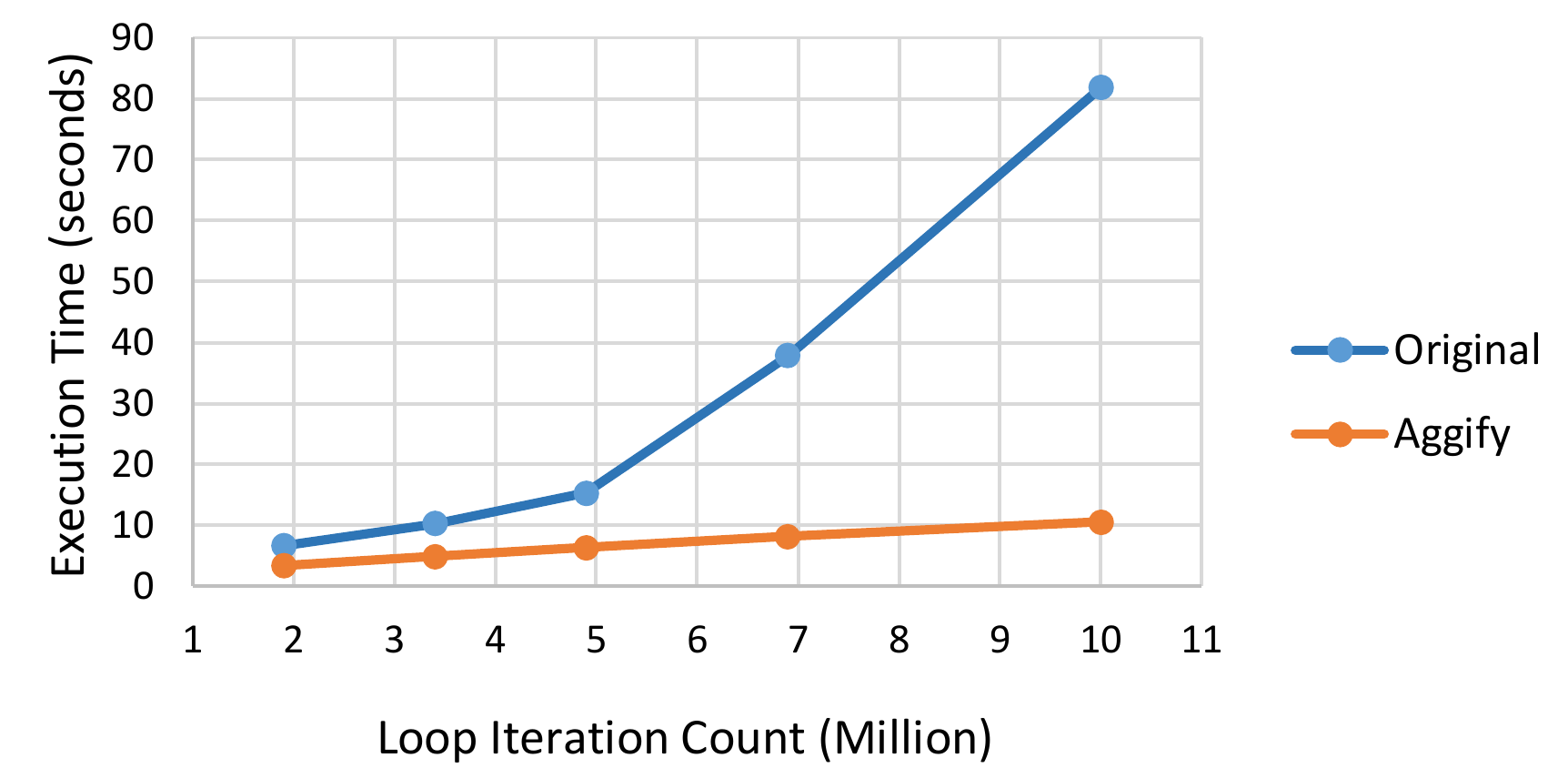}
{Scalability for Browse categories (Rubis Workload)}
{fig:rubis-scale}

% \onecolfigure
% {figs/eval/mcs-java-varycard10g}
% {Performance and data transferred with and without Aggify while varying loop iteration counts for MinCostSupplier (Java).}
% {fig:mcs-java-varycard10g}

% \onecolfigure
% {figs/eval/roi-vary-card}
% {Performance and data transferred with and without Aggify while varying loop iteration counts for CumulativeROI (Java).}
% {fig:roi-vary-card}

We now show the results of our experiments to evaluate the scalability of Aggify with varying data sizes. Figures \reffig{fig:scale}(a), (b), (c), \reffig{fig:tmw-scale} and \reffig{fig:rubis-scale} show the results for 5 experiments described below. The x-axis shows the number of loop iterations, and y-axis shows the execution time in in all the experiments.

\mysubsection{Experiment 1}
\reffig{fig:scale}(a) shows the results for TPC-H query Q2. We show the results for the original UDF (blue), the UDF transformed using Aggify (orange), and further applying Froid (red) indicated as \textit{Aggify+}. We observe that for smaller sizes, Aggify does not offer any improvement by itself. Beyond a certain point, the original program degrades drastically, while Aggify stays constant. For \textit{Aggify+}, we observe about an order of magnitude improvement in performance all through. 

\mysubsection{Experiment 2}
Next, we consider the Java implementation of minimum cost supplier functionality. The original program first runs a query that retrieves the required number of parts, and then the loop iterates over these parts, and computes the minimum cost supplier for each. We restrict the number of parts using a predicate on P\_PARTKEY. By varying its value, we can control the iteration count of the loop. There were 2 million tuples in the PART table, and hence we vary the iteration count from 200 to 2 million in multiples of 10. The transformed Java program eliminates this loop completely, and executes a query that makes use of the custom aggregate \textit{MinCostSuppAgg}. 

\reffig{fig:scale}(b) shows the results of this experiment with warm cache. The solid line (in blue) and the dashed line (in orange) represent the original program and the transformed program respectively. We observe that at smaller number of iterations, the benefits are lesser, but beyond 2K iterations, we see a consistent improvement by an order of magnitude. 

\mysubsection{Experiment 3}
We now consider a variant of the example described in \refsec{java-motivation} (\reffig{fig:jdbc-example}) that computes the cumulative rate of return on investments. The table had 50 columns which store the monthly rate of return per investment category for that investor. Here, we used the TOP keyword in SQL to control the iteration count of the loop. There were 3 million rows in the table; we varied the count from 30 to 3 million in multiples of 10. 
\reffig{fig:scale}(c) shows the results of this experiment. Similar to the earlier experiment, we see that beyond 3K iterations, Aggify starts to offer an order of magnitude improvement, all the way up to 3 million.

This transformation uses Aggify to eliminate the loop, and then follows the technique of ~\cite{Emani2016} as described in \refsec{subsec:integ}. Without Aggify, the technique of \cite{Emani2016} will not be able to translate this loop into SQL. The benefits for these two experiments are due to a combination of (i) pushing compute from the remote Java application into the DBMS, (ii) reducing the amount of data transferred from the DBMS to the application, and (iii) the SQL translation of ~\cite{Emani2016}.

\mysubsection{Experiment 4}
We consider loop L1 from the real workload W1 (the loop is given in~\cite{AGGIFYWL}) and vary loop iteration count; the results are given in \reffig{fig:tmw-scale}. The benefits of Aggify get better with scale, similar to the other scalability experiments. These benefits arise due to pipelining as well as reduction in data movement.

\mysubsection{Experiment 5}
Here, we evaluate the scalability of Aggify on the Rubis workload (Browse categories). The graph in figure \reffig{fig:rubis-scale} shows results similar to other scalability experiments.

\subsection{Data Movement}
One of the key benefits due to Aggify is the reduction in data movement from a remote DBMS to  client applications. We now measure the magnitude of data moved, and show how Aggify significantly reduces data movement. The results of this experiment for the MinCostSupplier and Cumulative ROI Java programs are plotted in Figures \ref{fig:scale}(b) and \ref{fig:scale}(c) using the secondary y-axis. In both figures, the dotted line (red) shows the data moved from the DBMS to the client for the original program in megabytes, and the dash-dot line (green) shows the data movement for the rewritten program. 

For the MinCostSupplier experiment, the original program ends up transferring $(140*n)$ bytes of data where $n$ is the number of iterations (i.e. number of parts), assuming 4-byte integers (P\_PARTKEY), 9-byte decimals(PS\_SUPPLYCOST) and 25-byte varchars (S\_NAME). The rewritten program transfers only $(38*n)$ bytes, resulting in a  reduction of 3.6x. For the CumulativeROI experiment, the original program transfers 200 bytes per iteration (assuming 4-byte floating point values). At 30 million tuples, this is 6GB of data transfer! Aggify only returns the result of the computation, since the entire loop is now executing inside the DBMS. This result is a single tuple with 50 floating point values (200 bytes) irrespective of the number of iterations. %From this experiment it is clear that Aggify can bring about huge benefits to database-backed applications.

\mysubsection{Overheads} Aggify performs the rewrite by making one pass over the cursor loop and the enclosing program fragment. We observe that this overhead is negligible in all our experiments.
\section{Related Work} \label{sec:relwork}
Optimization of loops in general, has been an active area of research in the compilers/PL community. Techniques for loop parallelization, tiling, fission, unrolling etc. are mature, and are part of state-of-the-art compilers~\cite{KENNEDYBOOK,MUCHNICK}. 
Lieuwen and DeWitt~\cite{Lieuwen92} describe techniques to optimize set iteration loops in object oriented database systems (OODBs). They show how to extend compilers to include database-style optimizations such as join reordering. 

There have been recent works that have explored the use of program synthesis to address problems such as (a) optimization of applications that use ORMs~\cite{Cheung13}, (b) translation of imperative programs into the Map Reduce paradigm~\cite{Radoi14, Ahmad18}. In contrast to these works, Aggify relies on program analysis and query rewriting. Further, Aggify expresses an entire loop as a relational aggregation operator. %Cheung et. al~\cite{Cheung13} aim to synthesize equivalent SQL from application code, while Radoi et. al~\cite{Radoi14, Ahmad18} show how to synthesize map and reduce functions.
Cheung et al.~\cite{CHEUNG12} show how to partition database application code such that part of the code runs inside the DBMS as a stored procedure. %They use application profiling and static analysis to identify this partitioning. 
Aggify also pushes computation into the DBMS, but moves entire cursor loops as an aggregate function thereby leveraging optimization techniques for aggregate functions~\cite{COH06}.

The idea of expressing loops as custom aggregates was first proposed by Simhadri et. al.~\cite{uudf} as part of the UDF decorrelation technique. Aggify is based on this idea. We (i) formally characterize the class of cursor loops that can be transformed into custom aggregates, (ii) relax a pre-condition given in ~\cite{uudf} thereby expanding the applicability of this technique, and (iii) show how this technique extends to FOR loops and applications that run outside the DBMS.

The DBridge line of work~\cite{GUR08, Emani2016, cobra18} has had many contributions in the area of optimizing data access in database applications using static analysis and query rewriting. \cite{GUR08} consider the problem of rewriting loops to make use of parameter batching. Emani et. al~\cite{Emani2016} describe a technique to translate imperative code to equivalent SQL. Recently, there have been efforts to optimize UDFs by transforming them into sub-queries or recursive CTEs~\cite{FroidVldb, Grust2019}. 
As we have shown in \refsec{subsec:integ}, Aggify can be used in conjunction with all these techniques leading to better performance.

%\cite{sqloop18}

%%%%%%%%%%%%%%%%%%%%%%%%%%%%%%%%%%
%\cite{Rheinlander2017}

%\cite{Cheung13}
%\cite{Radoi14}
%\cite{Ahmad18}

%\cite{uudf}
%\cite{Emani2016, cobra18}
%\cite{Grust2019}
%\cite{COH06}
%\cite{FroidVldb}
%Citation: \cite{Gal01}

% the following ways (i) by characterizing of the class of cursor loops that can be transformed into custom aggregates, (ii) relaxing the pre-conditions in ~\cite{uudf} thereby significantly expanding the applicability of this technique, and (iii) demonstrating how we can build parallelizable aggregates thereby resulting in better plans. 
%(iii) they cannot handle multiple reaching definitions, (iv) we dont require the values of init variables to be statically determinable.

%Notion of batch safety is given in Gur08

\section{Conclusion} \label{sec:concl}
Although it is well-known that set-oriented operations are generally more efficient compared to row-by-row operations, there are several scenarios where cursor loops are preferred, or are even inevitable. However, due to many reasons that we detail in this paper, cursor loops can result not only in poor performance, but also affect concurrency and resource consumption. Aggify, the technique presented in this paper, addresses this problem by automatically replacing cursor loops with SQL queries that invoke a custom aggregates that are systematically constructed based on the loop body. It performs this transformation while guaranteeing that the semantics of the loop are preserved. Our evaluation on benchmarks and real workloads show the potential benefits of such a technique. We believe that Aggify, can make a strong positive impact on real-world workloads both in database-backed applications as well as UDFs and stored procedures.

%% Add acknowledgements to Venkatesh

\bibliographystyle{ACM-Reference-Format}
\bibliography{aggify}

\end{document}